\newif\ifIEEEFORMAT
\IEEEFORMATfalse

\ifIEEEFORMAT
  \documentclass[letterpaper,10pt,conference]{ieeeconf}
  \IEEEoverridecommandlockouts
\else
  \documentclass[10pt,journal]{IEEEtran}
\fi

\usepackage[T1]{fontenc}
\usepackage[utf8]{inputenc}

\ifIEEEFORMAT
  \usepackage{amsmath,amssymb,amsfonts}

  \usepackage{amsthm}
\else
  \usepackage{amsmath,amssymb,amsfonts,amsthm}
\fi

\usepackage{newtxtext,newtxmath}
\usepackage{microtype}
\usepackage{bm}
\usepackage[hidelinks]{hyperref}
\ifIEEEFORMAT
\else
  \usepackage{titlesec}
\fi
\usepackage{cite}
\usepackage{graphicx}
\usepackage[dvipsnames]{xcolor}
\usepackage{tikz}
\usetikzlibrary{matrix,shapes.geometric,positioning,fit,backgrounds,calc}
\usepackage{array}
\usepackage{etoolbox}
\usepackage{environ}
\usepackage{comment}
\usepackage{float}
\usepackage{siunitx}
\usepackage{mathtools}
\usepackage{booktabs,tabularx,makecell}

\let\vv\relax

\usepackage{bm}
\usepackage{amsmath}
\usepackage{amssymb}

\newcommand{\vect}[1]{\mathbf{#1}}
\newcommand{\mat}[1]{\mathbf{#1}}

\newcommand{\diffs}[3]{\frac{\partial^2 #1}{
\ifx#2#3 
\partial #2^2
\else
\partial #2 \partial #3
\fi
}}

\newcommand{\bv}{\vect{b}}

\newcommand{\ev}{\vect{e}}

\newcommand{\fv}{\vect{f}}
\newcommand{\gv}{\vect{g}}

\newcommand{\hv}{\vect{h}}
\newcommand{\kv}{\vect{k}}

\newcommand{\pv}{\vect{p}}

\newcommand{\qv}{{\vect{q}}}

\newcommand{\rv}{{\vect{r}}}

\newcommand{\sv}{\vect{s}}

\newcommand{\uv}{\vect{u}}
\newcommand{\vv}{\vect{v}}

\newcommand{\wv}{\vect{w}}

\newcommand{\xv}{\vect{x}}

\newcommand{\yv}{\vect{y}}

\newcommand{\SO}{\mathbf{SO}(3)}

\newcommand{\rhov}{\bm{\rho}}

\newcommand{\Omegav}{\bm{\Omega}}

\newcommand{\rank}{\operatorname{rank}}

\newcommand{\Phim}{\bm{\Phi}}

\newcommand{\tauv}{\bm{\tau}}

\newcommand{\Gammam}{\bm{\Gamma}}

\newcommand{\Am}{\mat{A}}
\newcommand{\Bm}{\mat{B}}

\newcommand{\Dm}{\mat{D}}

\newcommand{\Jm}{\mat{J}}

\newcommand{\Km}{\mat{K}}

\newcommand{\Pm}{\mat{P}}

\newcommand{\Rm}{\mat{R}}

\newcommand{\Tm}{\mat{T}}

\newcommand{\Wm}{\mat{W}}

\newcommand{\PaperAuthorsPlain}
  {Mirko Mizzoni, Pieter van Goor, Barbara Bazzana, and Antonio Franchi}

\newcommand{\PaperAuthorsAffiliated}
  {Mirko Mizzoni$^{1}$, Pieter van Goor$^{3}$, Barbara Bazzana$^{1}$, and Antonio Franchi$^{1,2}$}

\newcommand{\AffiliationOne}{%
  $^1$ Robotics and Mechatronics group, Faculty of Electrical Engineering,
  Mathematics, and Computer Science (EEMCS), University of Twente,
  7500 AE Enschede, The Netherlands,
  \texttt{m.mizzoni@utwente.nl, b.bazzana@utwente.nl},
  \texttt{schol@r-franchi.eu}.%
}

\newcommand{\AffiliationTwo}{%
  $^2$ Department of Computer, Control and Management Engineering,
  Sapienza University of Rome, 00185 Rome, Italy,
  \texttt{schol@r-franchi.eu}.%
}

\newcommand{\AffiliationThree}{%
  $^3$ School of Aerospace, Mechanical, and Mechatronic Engineering (AMME),
  Faculty of Engineering, University of Sydney, NSW 2006, Australia,
  \texttt{pieter.vangoor@sydney.edu.au}.%
}

\newcommand{\FundingStatement}{%
  This work was partially funded by the Horizon Europe research agreement
  No.~101120732 (AUTOASSESS).%
}

\ifIEEEFORMAT
\else
  \renewcommand{\thesection}{\arabic{section}}
  \renewcommand{\thesubsection}{\thesection.\arabic{subsection}}

  \titleformat{\section}
    {\normalfont\large\bfseries}{\thesection}{0.65em}{}
  \titleformat{\subsection}
    {\normalfont\normalsize\bfseries}{\thesubsection}{0.65em}{}

  \titlespacing*{\section}
    {0pt}{1.8ex plus 0.4ex minus 0.2ex}{0.6ex}
  \titlespacing*{\subsection}
    {0pt}{1.25ex plus 0.3ex minus 0.2ex}{0.4ex}

  \makeatletter
  \def\maketitle{%
    \twocolumn[{%
      \raggedright
      {\LARGE\bfseries \@title\par}
      \vspace{0.7em}
      {\normalsize \@author\par}
      \vspace{0.55em}
      {\footnotesize
        \AffiliationOne\par
        \AffiliationTwo\par
        \AffiliationThree\par
        \FundingStatement\par
      }
      \vspace{1.0em}
    }]%
  }
  \makeatother
\fi

\theoremstyle{plain}

\newtheorem{exmp}{\textbf{Example}}[section]

\newtheorem{thm}{Theorem}
\newtheorem{lem}{Lemma}

\newtheorem{defn}{Definition}
\newtheorem{prop}{Proposition}

\newtheorem{rem}{\textbf{Remark}}[section]

\newcommand{\R}{\mathbb{R}}
\newcommand{\Exc}[1]{\!\{\,#1=0\,\}}

\newtoggle{omitproofs}
\togglefalse{omitproofs}
\NewEnviron{optproof}[1][]{%
  \iftoggle{omitproofs}{%
    \begin{proof}
    \textit{The proof \ifstrempty{#1}{}{can be found in~\cite{xx}}.}
    \vspace{-0.25em}
    \end{proof}
  }{%
    \begin{proof}
    \BODY
    \end{proof}
  }%
}

\newif\ifarxiv
\arxivtrue
\newcommand{\arxivtext}[1]{\ifarxiv #1\fi}

\title{\huge Input Dexterity and Output Negotiation in\\
Feedback-Linearizable Nonlinear Systems}

\ifIEEEFORMAT
  \author{\PaperAuthorsPlain%
    \thanks{\AffiliationOne}%
    \thanks{\AffiliationTwo}%
    \thanks{\AffiliationThree}%
    \thanks{\FundingStatement}%
  }
\else
  \author{\PaperAuthorsAffiliated}
\fi

\begin{document}

\maketitle
\togglefalse{omitproofs}
\arxivtrue

\begin{abstract}
We introduce a task-relative taxonomy of actuator inputs for nonlinear systems within the input-output feedback-linearization framework. Given a flat output specifying the task, inputs are classified as \emph{essential}, \emph{redundant}, or \emph{dexterity}: essential inputs are required for exact linearization, redundant inputs can be removed without effect, and dexterity inputs can be deactivated while preserving exact linearization of a reduced task.  We show that a subset is dexterity if and only if, under a suitable dynamic \emph{prolongation}, it can appear as additional output channels (\emph{flat-input complement}) on a common validity set. Whenever a family of systems obtained by (de)activating dexterity inputs admits a common prolongation, the family can be interpreted as a single prolonged system endowed with different output selections. This enables a unified linearizing controller that negotiates between full and reduced tasks without transients on shared outputs under compatibility and dwell-time conditions. 
Simulations on a fully actuated aerial platform illustrate graceful task downgrades from six-dimensional pose tracking as lateral-force channels are deactivated.
\end{abstract}

\section{Introduction}
Redundancy---having more independent actuation channels than minimally needed to regulate the outputs of interest---is a central motif across modern control applications. In input redundant nonlinear systems, an established architectural response is {control allocation}: a high-level controller specifies virtual inputs (e.g., forces/moments for mechanical systems), which are then distributed to redundant actuators while honoring constraints, handling saturations, and often providing fault tolerance~ \cite{Isidori1995}. Mature surveys and monographs show how control allocation underpins practice in {aerospace}, {marine/surface}, and {automotive} systems, among others~ \cite{JohansenFossen_Automatica2013,Fossen_Book2011}. In underwater and surface vessels, redundancy is the norm and allocation is formulated as constrained optimization; real-time and explicit solutions are now well documented~ \cite{FossenJohansen_UnderwaterSurvey2008,FossenJohansen_MED2006}. In ground vehicles, overactuation via brake/drive/steer subsystems motivates allocation methods (including MPC-based allocation) for stability and path-tracking at the limits~ \cite{VermillionCDC2007,EnergyEffAllocEV2010}. In aircraft and spacecraft, reconfigurable allocation is a workhorse for fault-tolerant flight/attitude control~ \cite{AlmutairiAouf_AeroJ2017,ShenAutomatica2015}. Multi-agent settings also blend redundancy and reconfiguration to maintain consensus or formation under actuator/sensor faults and cyber-physical disruptions~ \cite{DengJAS2020,Wang_SciChinaTechSci2021}.

Yet, the prevailing allocation viewpoint remains {controller-external}: redundancy is exploited to realize a {fixed} regulated output while optimizing secondary objectives and handling faults. What it does \emph{not} typically provide is a structural answer to a paramount task-engineering question: 
\begin{quote}
``When an actuator (or a set thereof) must be intentionally deactivated, which {subset of the task} can still be achieved under exact input-output linearization, and how can one switch to that reduced task {without transients} on the components that remain regulated?''
\end{quote}
 Existing degraded-mode designs (e.g., yaw-free flight, relaxed hover) specify ad-hoc reduced tasks and then build bespoke controllers for the post-fault plant~ \cite{MuellerDA_IJRR2016}.

\subsubsection*{Goal and viewpoint}
We develop a {task-relative} taxonomy of actuator inputs for nonlinear input-affine systems within the input-output feedback-linearization framework~ \cite{Isidori1995}. The task is encoded by a flat output. Relative to this task, we classify each input as:
\begin{itemize}
\item \emph{redundant}: its removal does not affect exact linearizability of the full task;
\item \emph{essential}: its removal precludes exact linearization of any nontrivial subset of the task;
\item \emph{dexterity}: its removal {still} permits exact linearization of a lower-dimensional task.
\end{itemize}
This classification is explicitly {output-dependent} (distinct from structural controllability) and directly answers: which parts of the task remain achievable under loss/deactivation of specific inputs?

\subsubsection*{Relation to state of the art}
Classical control allocation (surveyed in~ \cite{JohansenFossen_Automatica2013,Fossen_Book2011}) assumes a fixed task and distributes effort among redundant actuators---possibly dynamically (cf.\  dynamic allocation and ``weak/strong input redundancy'' in~ \cite{Zaccarian_Automatica2009})---but it does not {characterize} when an actuator set can be removed while still preserving exact linearizability of a {subtask}, nor does it provide a structural mechanism to {swap} inputs with outputs in the regulated vector. 

Conversely, the literature on {flat inputs} (dual to flat outputs) shows how to {place} inputs to render a system flat~ \cite{NicolauRespondekBarbot_SIAMJCO2020}, but it does not address actuator {removal} in a given plant, nor the synthesis of {zero-transient} switching among multiple flat outputs of a {common} prolonged system.

Recent work redefines {input redundancy} (IR) at the signal level: a system is IR if distinct inputs can yield the {same} output from the same initial state; LTI characterizations show this notion is {not} equivalent to classical null-space/allocator views of over-actuation, and IR may be {destroyed} by input/state constraints~ \cite{KreissTregouet_SCL2021,TregouetKreiss_Automatica2024}. 
By contrast, our framework is {task-relative} to a chosen flat output: we classify inputs as {redundant/essential/dexterity} by whether exact input-output linearization survives under input removal after a suitable {dynamic prolongation} and by enabling input-output {negotiation}~ \cite{Isidori1995}. 
We further prove a structural equivalence---{dexterity} $\Leftrightarrow$ existence of a {flat-input complement} under a {common prolongation}---and exploit it for {zero-transient} switching among compatible tasks; IR does not address task negotiation or meld switching~ \cite{KreissTregouet_SCL2021,TregouetKreiss_Automatica2024}.

\subsubsection*{Contributions}
\emph{(C1) Dexterity $\iff$ flat-input complement under prolongation.}
We prove that a subset of inputs is \emph{dexterity} if and only if, on a nonempty intersection of validity sets, there exists a \emph{dynamic prolongation} $\ell$ for which the same plant admits a flat output where those inputs appear as channels in the output (\emph{flat-input complement}). This gives a formal \emph{negotiation} mechanism: selected inputs can be exchanged with components of the original output while preserving exact linearizability.  
\emph{(C2) A common prolongation unifies actuation modes.}
Whenever the family of systems obtained by (de)activating dexterity inputs admits a \emph{common prolongation}, it can be interpreted as a single prolonged system endowed with different flat-output selections. Unlike allocation (which preserves a fixed task), this reinterprets ``fully actuated vs.\ under-actuated modes'' as different outputs of the same prolonged plant.

\emph{(C3) Zero-transient switching and application.}
Building on our prior result on \emph{meld} switching~\cite{mizzoni2026switchingfeedbacklinearizingoutputsets}, we show that, under a common prolongation, one can switch between the original task and reduced tasks {without transients} on their shared components, provided the outputs are compatible on an overlapping validity set and a dwell-time condition holds.

\subsubsection*{Organization}
Section~\ref{sec:motv_exmp} gives the motivating example and Section~\ref{sec:example_observations} distills terminology; Section~\ref{subs:preliminaries} recalls feedback-linearization and  Section~\ref{sec:notation} introduces notation; Section~\ref{sec:io-class} defines the input taxonomy, and Section~\ref{sec:FIC-equivalence} proves the dexterity-flat-input-complement equivalence; Secs.~\ref{sec:negotiation-control}--\ref{sec:switching-negotiable-common-ell} develop negotiation control, common prolongation, and switching guarantees; Section~\ref{sec:rigid-body} presents the rigid-body case study; Section~\ref{sec:mec_wheel} discusses the Mecanum-wheel example, Section ~\ref{sec:ua-fa} compares the two scenarios and highlights the common interpretation.
Finally, Section~\ref{subs:conclusion} concludes the work.

\section{Motivating Example}\label{sec:motv_exmp}

Consider the linear input-affine system
\begin{equation}
\widetilde{\Sigma}:
\left\{
\begin{split}
\dot{x}_1&=x_2,\\
\dot{x}_2&=x_3+u_1+u_2,\\
\dot{x}_3&=x_4+u_2,\\
\dot{x}_4&=u_3+u_4
\end{split}
\right.
\end{equation}
with $n=4$ states and four inputs. Select first the $4$-dimensional output
$\widetilde{\mathbf y}=\mathbf x$. The associated decoupling matrix
\[
\widetilde{\Am}=
\begin{bmatrix}
1 & 1 & 0 & 0\\
1 & 1 & 0 & 0\\
0 & 1 & 0 & 0\\
0 & 0 & 1 & 1
\end{bmatrix}
\]
is singular, hence exact input-output linearization with a $4$-dimensional output is not possible.

We therefore consider a $3$-dimensional output
\[
\mathbf y=\begin{bmatrix}x_1 & x_3 & x_4\end{bmatrix}^\top.
\]
Its decoupling matrix
\[
\overline{\mathbf A}=
\begin{bmatrix}
1 & 1 & 0 & 0\\
0 & 1 & 0 & 0\\
0 & 0 & 1 & 1
\end{bmatrix}
\]
has rank $3$ and the sum of relative degrees is $|{\mathbf r}|=2+1+1=4=n$; thus $\mathbf y$ is a flat output for $\widetilde\Sigma$ (no zero dynamics).
Note that the third and fourth columns of $\overline{\mathbf A}$ coincide, i.e., $[\;0\;0\;1\;]^\top$, so one of $u_3,u_4$ is redundant for this choice of $\mathbf y$.
 Remove $u_4$ and consider the resulting \emph{square} system
\begin{equation}
\Sigma:
\left\{
\begin{split}
\dot{x}_1 &= x_2,\\
\dot{x}_2 &= x_3 + u_1 + u_2,\\
\dot{x}_3 &= x_4 + u_2,\\
\dot{x}_4 &= u_3,
\end{split}\right.
\label{eq:square_system}
\end{equation}
for which the same output $\mathbf y=[x_1\;x_3\;x_4]^\top$ remains flat, with
\(
{\mathbf A}=
\begin{bmatrix}
1 & 1 & 0\\
0 & 1 & 0\\
0 & 0 & 1
\end{bmatrix}
\)
invertible and $|{\mathbf r}|=4$.

We now ask what happens if we \emph{remove further inputs} from the square system \(\Sigma\).

\smallskip
\noindent\emph{Single removals.}
\begin{itemize}
\item Removing $u_3$ (\(\Sigma_{\overline{\{3\}}}\)): no two-dimensional subset of entries of $\mathbf y$ achieves $|{\mathbf r}|=4$; exact linearization with no zero dynamics is impossible.
\item Removing $u_2$ (\(\Sigma_{\overline{\{2\}}}\)): the subset $\begin{bmatrix}x_1&x_3\end{bmatrix}^\top$ is flat; its decoupling matrix is the identity and $|{\mathbf r}|=2+2=4$.
\item Removing $u_1$ (\(\Sigma_{\overline{\{1\}}}\)): again, no two-dimensional subset of entries of $\mathbf y$ reaches $|{\mathbf r}|=4$.
\end{itemize}

\noindent\emph{Double removals.}
\begin{itemize}
\item Removing $\{u_1,u_2\}$ (\(\Sigma_{\overline{\{1,2\}}}\)): the scalar output $y_1=x_1$ has relative degree $4$ and is flat (a single chain of four integrators to $u_3$).
\item Removing $\{u_1,u_3\}$ or $\{u_2,u_3\}$: the state $x_4$ becomes uncontrollable ($\dot{x}_4=0$), hence no scalar flat output of the form $y_i\in\{x_1,x_3,x_4\}$ exists.
\end{itemize}

This example will serve as a compact guide for the notions introduced in the next section.

\section{Immediate Observations and Terminology}\label{sec:example_observations}

The example above highlights four distinct phenomena that our theory will formalize.

\subsubsection*{(O1) Redundancy in a rectangular setting}
For the rectangular system \(\widetilde\Sigma\) with the $3$-dimensional output \(\mathbf y=[x_1\;x_3\;x_4]^\top\), the last two columns of the decoupling matrix coincide; therefore any of the inputs \((u_3,u_4)\) is \emph{redundant} for this output (we removed, e.g., \(u_4\) to obtain the square system \(\Sigma\)).

\subsubsection*{(O2) Dexterity inputs (loss degree one)}
In the square system \(\Sigma\), removing \(u_2\) still permits exact linearization of a lower-dimensional \emph{subset} of the original output, namely \([x_1\;x_3]^\top\) (two outputs instead of three). We call \(u_2\) a \emph{dexterity input of loss degree one} for \((\Sigma,\mathbf y)\): the controllable output dimension decreases by one upon its removal, yet exact input-output linearization (with no zero dynamics) is retained for the surviving subset.

\subsubsection*{(O3) Dexterity subsets (higher loss degrees)}
Removing the \emph{pair} \(\{u_1,u_2\}\) yields a system in which \(x_1\) alone is a flat output (loss of two output entries). Hence \(\{u_1,u_2\}\) is a \emph{dexterity subset of loss degree two}. Note that \(u_1\) by itself is \emph{not} a loss-one dexterity input (its removal alone does not preserve a two-dimensional flat output), but it belongs to a dexterity subset of cardinality two.

\subsubsection*{(O4) Essential inputs}
Input \(u_3\) is \emph{essential} for \((\Sigma,\mathbf y)\): every time it is removed, exact linearization with output subsets of \(\mathbf y\) fails (either the sum of relative degrees falls short of the state dimension, or controllability is lost).

\medskip
In summary:
\begin{itemize}
\item Redundancy (in a rectangular setting) does not imply dexterity in the square setting.
\item Dexterity is a \emph{task-relative} notion: it is defined with respect to a given flat output \(\mathbf y\) and captures whether exact linearization survives input removals by allowing a \emph{reduced} output.
\item The \emph{loss degree} of an input is the minimal cardinality of a dexterity subset containing it (here: \(\deg(u_2)=1\), \(\deg(u_1)=2\), \(u_3\) is essential).
\end{itemize}

These observations motivate the formal definitions (dexterity subset, loss degree) and-crucially-the nonlinear developments in the remainder of the paper. In particular, we will show that in the nonlinear case the existence of dexterity subsets is \emph{equivalent} to the existence of a suitable \emph{common prolongation} (dynamic extension) under which one can negotiate components of the original output with selected input channels, and that this perspective enables \emph{zero-transient} switching between full and reduced tasks under a single feedback-linearizing controller.

\section{Preliminaries}\label{subs:preliminaries}

We recall standard notions from input-output feedback linearization and fix notation; see~ \cite[Ch.~5]{Isidori1995} for comprehensive background. We also summarize, in a self-contained form, the specific switching fact from our prior work on \emph{melds} that will be used later; full details are in~ \cite{mizzoni2026switchingfeedbacklinearizingoutputsets}.

\subsubsection*{System model and outputs}
Throughout, statements and constructions are local on an open set $\mathcal{X}\subset\mathbb{R}^n$; all maps are $C^\infty$ unless otherwise stated.
Consider the smooth, input-affine, multivariable system
\begin{equation}
  \Sigma:\quad \dot{\xv} = \fv(\xv) + \sum_{i=1}^p \gv_i(\xv)\,u_i,
  \qquad \xv\in\R^n,\ \uv\in\R^p,
  \label{eq:sys}
\end{equation}
with smooth vector fields $\fv,\gv_i:\mathcal{X}\to\R^n$ defined on an open set $\mathcal{X}\subset\R^n$. 
Let the output be
\[
  \yv=\hv(\xv)=\begin{bmatrix}h_1(\xv)&\cdots&h_p(\xv)\end{bmatrix}^\top,\qquad 
  \hv:\mathcal{X}\to\R^p \text{ smooth}.
\]

\subsubsection*{Lie derivatives and vector relative degree}
For any smooth $h:\mathcal{X}\to\R$, denote by $L_{\fv} h$ and $L_{\gv_j}h$ the Lie derivatives of $h$ along $\fv$ and $\gv_j$, and by $L_{\fv}^k h$ the $k$-th iterated Lie derivative along $\fv$. The pair $(\Sigma,\yv)$ is said to have \emph{vector relative degree} 
\(
\rv=\begin{bmatrix}r_1&\cdots&r_p\end{bmatrix}^\top
\)
at $\xv^\circ\in\mathcal{X}$ if, for each $i\in\{1,\dots,p\}$:
\begin{itemize}
  \item $L_{\gv_j} L_{\fv}^{k} h_i(\xv)=0$ for all $j\in\{1,\dots,p\}$ and all $k<r_i-1$ in an open neighborhood of $\xv^\circ$ where the vector relative degree is constant;
  \item the \emph{decoupling matrix}
  \begin{equation}
    \Am_{(\Sigma,\yv^{(\rv)})}(\xv)
     \;:=\; \Big[\, L_{\gv_j} L_{\fv}^{\,r_i-1} h_i(\xv)\,\Big]_{i,j=1}^p,
     \label{eq:decoupling}
  \end{equation}
  is nonsingular at $\xv^\circ$.
\end{itemize}
If \eqref{eq:decoupling} holds, then the array of output derivatives at orders $\rv$ satisfies
\begin{equation}
  \yv^{(\rv)} 
  := \begin{bmatrix} y_1^{(r_1)} & \cdots & y_p^{(r_p)} \end{bmatrix}^\top
  \;=\; \bv(\xv) + \Am_{(\Sigma,\yv^{(\rv)})}(\xv)\,\uv,
  \label{eq:y_r_now}
\end{equation}
with
\(
  \bv(\xv):=\begin{bmatrix}
   L_{\fv}^{r_1}h_1(\xv)&\cdots&L_{\fv}^{r_p}h_p(\xv)
  \end{bmatrix}^\top.
\)

\subsubsection*{Exact state-space linearization (via the output)}
Suppose $(\Sigma,\yv)$ has vector relative degree $\rv$ at $\xv^\circ$ and that
\(
  |\rv|:=\sum_{i=1}^p r_i = n.
\)
On any open set where $\Am_{(\Sigma,\yv^{(\rv)})}(\xv)$ is nonsingular and $\rv$ is constant, the static state feedback
\[
  \uv \;=\; \Am_{(\Sigma,\yv^{(\rv)})}^{-1}(\xv)\,\big(\,-\bv(\xv)+\vv\,\big), 
  \qquad \vv\in\R^p,
\]
transforms \eqref{eq:y_r_now} into decoupled chains of integrators,
\(
 \yv^{(\rv)}=\vv,
\)
i.e., the system is (locally) put into Brunovsk\'{y} form with no zero dynamics~ \cite{Isidori1995}.

\subsubsection*{Validity set of an output}
Fix $\xv^\circ\in\mathcal{X}$. The \emph{validity set} of the output $\yv$ at $\xv^\circ$ is the largest set
\[
  \mathcal{B}(\xv^\circ)
  :=\Big\{\xv\in\mathcal{X}\;:\;
   \rank \Am_{(\Sigma,\yv^{(\rv)})}(\xv)=p\ \wedge\ 
   \rv(\xv)=\rv(\xv^\circ)\Big\},
\]
i.e., the set where the vector relative degree is well defined and constant and the decoupling matrix is nonsingular. On $\mathcal{B}(\xv^\circ)$, the feedback above yields exact state-space linearization. In the sequel, we will use the term \emph{flat output} for such an output $\yv$ (equivalently, an exact state-space linearizing output; cf.~ \cite{Isidori1995}).
We write $\mathcal{B}(\xv^\circ;\Sigma,\yv)$ if the dependence on the pair $(\Sigma,\yv)$ needs to be emphasized; for prolonged systems we analogously use $\mathcal{B}_\ell(\xv^\circ;\Sigma^{(\ell)},\cdot)$.

\subsubsection*{Compatibility of outputs}
Let $\yv_1$ and $\yv_2$ be flat outputs for $\Sigma$ with validity sets $\mathcal{B}_1(\xv^\circ_1)$ and $\mathcal{B}_2(\xv^\circ_2)$. They are said to be \emph{compatible} if there exists a point $\xv^\ast$ and an open neighborhood $\mathcal{U}\ni\xv^\ast$ with
\(
  \mathcal{U}\subseteq \mathcal{B}_1(\xv^\circ_1)\cap \mathcal{B}_2(\xv^\circ_2).
\)
Intuitively, both linearizing feedbacks are well posed and the relative-degree vectors are constant on the same open set, so switching between them is feasible.

\subsubsection*{A brief note on meld switching (used later)}
We will occasionally work with a finite \emph{deck} of candidate scalar output maps and with \emph{melds}, i.e., selections of $p$ scalar output maps from the deck that are flat outputs for the same system. The specific property we need-stated here without proof-is the following consequence of~ \cite{mizzoni2026switchingfeedbacklinearizingoutputsets}:  
(i) if two flat outputs have a nonempty intersection of their validity sets, and \emph{share} a subset of components, then under the same exact-linearizing feedback the closed-loop tracking error dynamics of the shared components are unchanged across a switch;  
(ii) switching among compatible outputs with a suitable dwell-time plus a suitable margin conditions, keeps the closed-loop state bounded.  
We will apply (i) in Sections~\ref{sec:common-prolongation}--\ref{sec:switching-negotiable-common-ell} to guarantee \emph{no transients} on shared outputs when negotiating between full and reduced tasks under a common prolongation; (ii) is only needed for completeness of the switching argument.

\section{Notation}\label{sec:notation}

\subsubsection*{Index sets and slicing}
Let $\mathcal{I}:=\{1,\ldots,p\}$ be the index set for input and output entries. Its power set is $\mathcal{P}(\mathcal{I})$. For any $\mathcal{S}\subset\mathcal{I}$, write $\overline{\mathcal{S}}:=\mathcal{I}\setminus\mathcal{S}$ for the complement.  
Given a set $\mathbb{X}$ and any array $\qv\in\mathbb{X}^p$, define the subvectors
$\qv_{\mathcal{S}}\in\mathbb{X}^{|\mathcal{S}|}$ and $\qv_{\overline{\mathcal{S}}}\in\mathbb{X}^{p-|\mathcal{S}|}$ by extracting the components of $\qv$ whose indices lie in $\mathcal{S}$ and $\overline{\mathcal{S}}$, respectively, preserving the natural order induced by $\mathcal{I}$.  
For $\qv',\qv''\in\mathbb{X}^p$ and $\mathcal{S}\subset\mathcal{I}$, define the \emph{merge} operator
$$
\operatorname{merge}_{\mathcal{S}}(\qv',\qv'') \in \mathbb{X}^p,
$$
as the array obtained by taking entries from $\qv'$ on the indices in $\mathcal{S}$ and from $\qv''$ on the indices in $\overline{\mathcal{S}}$.

\subsubsection*{Input removal and reduced systems}
For $\mathcal{A}\subseteq\mathcal{I}$, denote by $\Sigma_{\overline{\mathcal{A}}}$ the system obtained from $\Sigma$ by \emph{keeping only} the inputs with indices in $\overline{\mathcal{A}}$ (equivalently, by setting $u_i\equiv 0$ for $i\in\mathcal{A}$), i.e.,
\[
\Sigma_{\overline{\mathcal{A}}}:\quad
\dot{\xv}=\fv(\xv)+\sum_{i\in\overline{\mathcal{A}}}\gv_i(\xv)\,u_i.
\]

\subsubsection*{Prolongation (dynamic extension)}
For $\mathcal{A}\subseteq\mathcal{I}$, define the restricted nonnegative lattice
\[
\mathbb{N}_0^{p}\big|_{\overline{\mathcal{A}}}
:=\big\{\ell=(l_1,\ldots,l_p)\in\mathbb{N}_0^{p}\;:\;l_i=0\ \text{for all }i\in\mathcal{A}\big\}.
\]
An element $\ell\in\mathbb{N}_0^{p}\big|_{\overline{\mathcal{A}}}$ is a \emph{prolongation pattern} (dynamic-extension pattern). Let
\[
|\ell|:=\sum_{i\in\overline{\mathcal{A}}} l_i,
\]
and define $\Sigma^{(\ell)}_{\overline{\mathcal{A}}}$ as the system obtained by interconnecting $\Sigma_{\overline{\mathcal{A}}}$ with $p-|\mathcal{A}|$ chains of integrators of lengths $l_i$ on the surviving inputs:
\begin{equation}
\Sigma^{(\ell)}_{\overline{\mathcal{A}}}:\quad
\dot{\xv}=\fv(\xv)+\sum_{i\in\overline{\mathcal{A}}}\gv_i(\xv)\,u_i,\qquad
u_i^{(l_i)}=v_i\ \ \forall\,i\in\overline{\mathcal{A}}.
\label{eq:sigma_ext_minus_A}
\end{equation}
Here $u_i^{(k)}$ denotes the $k$-th time derivative of $u_i$, and $v_i$ is the corresponding \emph{virtual input}. The state dimension of $\Sigma^{(\ell)}_{\overline{\mathcal{A}}}$ is $n+|\ell|$. This construction is the standard \emph{dynamic extension} (or \emph{prolongation}) used in feedback linearization; see, e.g.,~ \cite{Isidori1995,CharletLevineMarino1991,doi:10.1080/00207178508933432} and recent developments on pure prolongation conditions~ \cite{LevinePureProlongation2023}.%
\footnote{In the literature, dynamic extension is also called \emph{dynamic compensation} or \emph{linearization by prolongation}; see~ \cite{doi:10.1137/0329002} and references therein.}
For any pattern $\ell\in\mathbb{N}_0^p$ and index set $\mathcal{S}\subset\mathcal{I}$, let $\ell_{\mathcal{S}}\in\mathbb{N}_0^{|\mathcal{S}|}$ denote the subarray restricted to $\mathcal{S}$ and $|\ell_{\mathcal{S}}|:=\sum_{i\in\mathcal{S}} l_i$.

\subsubsection*{Full prolongation and stacked variables}
When $\mathcal{A}=\emptyset$ (i.e., $\overline{\mathcal{A}}=\mathcal{I}$), we write simply $\Sigma^{(\ell)}:=\Sigma^{(\ell)}_{\overline{\emptyset}}$. A convenient state for $\Sigma^{(\ell)}$ is
\[
\xv_{\ell}
:=\big[\,\xv^\top,\ \underbrace{u_1,\dot u_1,\ldots,u_1^{(l_1-1)}}_{\text{$l_1$ terms}},\ \ldots,\ 
\underbrace{u_p,\dot u_p,\ldots,u_p^{(l_p-1)}}_{\text{$l_p$ terms}}\big]^\top
\in \R^{n_\ell},
\]
with $n_\ell=n+|\ell|$, and the virtual input
\(
\vv := \big[u_1^{(l_1)}\ \cdots\ u_p^{(l_p)}\big]^\top.
\)

\subsubsection*{Derivative arrays}
Given a time-dependent array $\sv(t)=[s_1(t)\ \cdots\ s_p(t)]^\top$ and a pattern $\ell\in\mathbb{N}_0^p$, define
\[
\sv^{(\ell)}
:=\big[s_1^{(l_1)}\ \cdots\ s_p^{(l_p)}\big]^\top,\qquad
\sv^{(\ell-1)}
:=\big[s_1^{(l_1-1)}\ \cdots\ s_p^{(l_p-1)}\big]^\top,
\]
with the convention that entries with $l_i=0$ are omitted in $\sv^{(\ell-1)}$ (and the feedback for that channel will have no error filter on the input stack.). This notation will be used for both input arrays and output arrays under prolongation.

\section{Dexterity Inputs and Their Characterization}\label{sec:io-class}

We propose an input classification framework relative to a given flat output. 
Consider the input-affine system $\Sigma$ in~\eqref{eq:sys} under the regularity assumptions of Section~\ref{subs:preliminaries}, with $p$ inputs and a $p$-dimensional output 
$\yv=[y_1\ \cdots\ y_p]^\top$. 
We study how the exact state-space linearizability (flatness) of $\yv$, or of \emph{subsets} of its entries, behaves when selected input channels are \emph{removed} (i.e., set identically to zero).

\subsubsection*{Dexterity subsets and loss}

\begin{defn}[Dexterity subset of inputs]\label{defn:dextsubset}
Let $\yv$ be a flat output for $\Sigma$ on a nonempty validity set $\mathcal{B}(\xv^\circ)$ (Section~\ref{subs:preliminaries}). 
A subset of input indices $\mathcal{A}\in\mathcal{P}(\mathcal{I})$ is a \emph{dexterity subset of inputs} for $(\Sigma,\yv)$ if there exist
\begin{enumerate}
  \item a prolongation pattern $\ell=(l_1,\ldots,l_p)\in \mathbb{N}_0^{p}\big|_{\overline{\mathcal{A}}}$, and
  \item an index set $\mathcal{O}\subset\mathcal{I}$ with $|\mathcal{O}|=|\mathcal{A}|$,
\end{enumerate}
such that the reduced output $\yv_{\overline{\mathcal{O}}}\in\R^{p-|\mathcal{A}|}$ 
is a flat output for the reduced prolonged system $\Sigma^{(\ell)}_{\overline{\mathcal{A}}}$ in the sense of Section~\ref{subs:preliminaries}, i.e., $|\rv|=n+|\ell_{\overline{\mathcal{A}}}|$ and $\Am_{(\Sigma^{(\ell)}_{\overline{\mathcal{A}}},\yv^{(\rv)}_{\overline{\mathcal{O}}})}$ is nonsingular on a nonempty validity set.
The integer $|\mathcal{A}|$ is the \emph{dexterity loss} associated with $\mathcal{A}$.
\end{defn}

Intuitively, removing the $|\mathcal{A}|$ input channels in $\mathcal{A}$ still allows exact input-output linearization of a \emph{reduced} set of $p-|\mathcal{A}|$ output components, possibly after a suitable dynamic extension on the surviving inputs~ \cite{Isidori1995,DescusseMoog1985}.

\begin{defn}[Family of dexterity subsets]\label{defn:Dsigma}
The \emph{family of dexterity subsets} for $(\Sigma,\yv)$ is
\[
\mathcal{D}_{(\Sigma,\yv)} \;:=\;
\Big\{\,\mathcal{A}\in\mathcal{P}(\mathcal{I})\;\big|\;\mathcal{A}\ \text{is a dexterity subset for }(\Sigma,\yv)\,\Big\}.
\]
For $i\in\mathcal{I}$, the subfamily containing $i$ is
\(
\mathcal{D}^i_{(\Sigma,\yv)} := \big\{\,\mathcal{A}\in\mathcal{D}_{(\Sigma,\yv)}:\ i\in\mathcal{A}\,\big\}.
\)
\end{defn}
The mapping $\yv\mapsto\mathcal{D}_{(\Sigma,\yv)}$ is output-dependent; two flat outputs of the same system may yield different families.

\begin{defn}[Essential vs.\ dexterity inputs; minimum loss]\label{defn_dext_ess}
An input $u_i$ is \emph{essential} for $(\Sigma,\yv)$ if $\mathcal{D}^i_{(\Sigma,\yv)}=\emptyset$. 
Otherwise, $u_i$ is a \emph{dexterity input}. 
For a dexterity input $u_i$, the \emph{minimum dexterity loss} is
\[
\delta^i_{(\Sigma,\yv)} \;:=\; \min_{\mathcal{A}\in \mathcal{D}^i_{(\Sigma,\yv)}} |\mathcal{A}|
\;\in\; \{1,\ldots,p-1\}.
\]
\end{defn}

\subsubsection*{Task-relative nature and relation to controllability}
\begin{rem}
Dexterity is \emph{relative} to the chosen flat output $\yv$. An input may be essential for $(\Sigma,\yv)$ yet belong to a dexterity subset for $(\Sigma,\overline{\yv})$ for a different flat output $\overline{\yv}$. 
Moreover, essentiality is distinct from (structural) controllability: an input being essential for $(\Sigma,\yv)$ does not, by itself, imply loss of system controllability when that channel is zeroed; it only says that \emph{exact linearization of the intended task} (full or reduced $\yv$) cannot be preserved under that removal~ \cite{Isidori1995}.
\end{rem}

\subsubsection*{Example: revisiting Section~\ref{sec:motv_exmp}}
For the square system $\Sigma$ and output $\yv=[x_1\ x_3\ x_4]^\top$ in Section~\ref{sec:motv_exmp}, one obtains
\[
\mathcal{D}_{(\Sigma,\yv)}=\big\{\{2\},\,\{1,2\}\big\}.
\]
Hence $u_1$ and $u_2$ are dexterity inputs with
\(
\mathcal{D}^1_{(\Sigma,\yv)}=\{\{1,2\}\},\ \delta^1_{(\Sigma,\yv)}=2
\)
and
\(
\mathcal{D}^2_{(\Sigma,\yv)}=\{\{2\},\{1,2\}\},\ \delta^2_{(\Sigma,\yv)}=1.
\)
Input $u_3$ is \emph{essential} since $\mathcal{D}^3_{(\Sigma,\yv)}=\emptyset$.

\subsubsection*{Non-closure under unions (and a concrete counterexample)}
The family $\mathcal{D}_{(\Sigma,\yv)}$ need not be closed under set union. 
Even if two inputs $u_i$ and $u_j$ are \emph{individually} loss-one dexterity inputs (i.e., $\{i\},\{j\}\in\mathcal{D}_{(\Sigma,\yv)}$), the pair $\{i,j\}$ may fail to be a dexterity subset.

\begin{exmp}\label{exmp:1-2-dext-12-no}
Consider
\begin{equation*}
\Sigma: \left\{
\begin{split}
\dot{x}_1&=x_2,\\
\dot{x}_2&=x_4+u_2,\\
\dot{x}_3&=x_4+u_1,\\
\dot{x}_4&=u_3,
\end{split}\right.
\end{equation*}
with $n=4$, $p=3$, and $\yv=[x_1\ x_3\ x_4]^\top$. 
Then $\yv$ is a flat output for $\Sigma$ (no zero dynamics). 
Moreover, $u_1$ and $u_2$ are each loss-one dexterity inputs: $\mathcal{D}^1_{(\Sigma,\yv)}=\{\{1\}\}$ and $\mathcal{D}^2_{(\Sigma,\yv)}=\{\{2\}\}$. 
However, $\{1,2\}\notin\mathcal{D}_{(\Sigma,\yv)}$ because for every admissible choice $\mathcal{O}\in\{\{1,2\},\{1,3\},\{2,3\}\}$
the reduced output $\yv_{\overline{\mathcal{O}}}$ has $\sum_i r_i < n+|\widetilde{\ell}_{\overline{\mathcal{A}}}|$ (i.e., total relative degree smaller than the state dimension of $\Sigma^{(\widetilde{\ell})}_{\overline{\mathcal{A}}}$), hence cannot be flat. 
 Input $u_3$ is essential since its removal makes $x_4$ uncontrollable. 
Thus $\mathcal{D}_{(\Sigma,\yv)}=\{\{1\},\{2\}\}$. 
\end{exmp}

\subsubsection*{Realizing pairs $(\ell,\mathcal O)$ and multiplicity}
Since the prolongation pattern and the output index set that realize Definition~\ref{defn:dextsubset} need not be unique, we collect all admissible pairs:
\begin{defn}[Realizing pairs]
For $\mathcal{A}\in\mathcal{D}_{(\Sigma,\yv)}$, define
\[
\overline{\mathcal{C}}^{\mathcal{A}}_{(\Sigma,\yv)} 
:=\Big\{\, (\ell,\mathcal{O})\in \mathbb{N}_0^{p}\big|_{\overline{\mathcal{A}}}\times\mathcal{P}(\mathcal{I})\ \Big|\ 
\yv_{\overline{\mathcal{O}}}\ \text{is flat for}\ \Sigma^{(\ell)}_{\overline{\mathcal{A}}}\,\Big\}.
\]
\end{defn}

\begin{rem}
Multiplicity can occur in both arguments: for a fixed $\mathcal{O}$ there may exist finitely or infinitely many patterns $\ell$ yielding flatness, and for a fixed $\ell$ several index sets $\mathcal{O}$ may work (cf. dynamic extension/prolongation and pure-prolongation characterizations)~ \cite{Isidori1995,CharletLevineMarino1991,LevinePureProlongation2025}. 
Concrete instances appear in our rigid-body example (Section~\ref{sec:rigid-body}).
\end{rem}

\medskip

While the preceding lemma  characterizes the algebraic transition between known dexterity configurations, identifying these sets directly from Definition~\ref{defn:dextsubset} still requires building and analyzing the $\overline{\mathcal{A}}$--reduced system $\Sigma_{\overline{\mathcal{A}}}$ (possibly after dynamic prolongation) for each candidate subset $\mathcal{A} \subset \mathcal{I}$. 
In the next section we derive an \emph{equivalent} condition phrased \emph{solely} on the original system $\Sigma$, providing both analytic simplification and structural insight.

\section{Flat-Input Complement and Equivalence with Dexterity}\label{sec:FIC-equivalence}

We now introduce the \emph{flat-input complement} viewpoint and show that, under a mild compatibility condition on the validity set, it is \emph{equivalent} to dexterity.

\begin{defn}[Flat-input complement subset]\label{defn:FIC}
Let $\Sigma$ be as in~\eqref{eq:sys}, and let $\yv\in\R^p$ be a flat output for $\Sigma$. A subset of input indices $\mathcal{A}\subset\mathcal{I}$ is a \emph{flat-input complement subset} for $(\Sigma,\yv)$ if there exist
\begin{enumerate}
  \item a prolongation pattern $\ell\in\mathbb{N}_0^p$, and
  \item an index set $\mathcal{O}\subset\mathcal{I}$ with $|\mathcal{O}|=|\mathcal{A}|$,
\end{enumerate}
such that the augmented output
\[
\yv_{\overline{\mathcal{O}},\mathcal{A}}
:=\begin{bmatrix}\yv_{\overline{\mathcal{O}}}\\ \uv_{\mathcal{A}}\end{bmatrix}\in\R^p
\]
is a flat output for the prolonged system $\Sigma^{(\ell)}$ (Section~\ref{subs:preliminaries}). 
\end{defn}

\begin{defn}[Family of flat-input complement subsets]\label{defn:Ffamily}
The family of flat-input complement subsets for $(\Sigma,\yv)$ is
\(
\mathcal{F}_{(\Sigma,\yv)}
:=\big\{\mathcal{A}\subset\mathcal{I}\ \big|\ \mathcal{A}\ \text{is a flat-input complement subset for }(\Sigma,\yv)\big\}.
\)
\end{defn}

\begin{defn}[Realizing pairs for flat-input complement]\label{defn:CA}
Given $\mathcal{A}\in\mathcal{F}_{(\Sigma,\yv)}$, define
\[
\mathcal{C}^{\mathcal{A}}_{(\Sigma,\yv)}
:=\Big\{\,(\ell,\mathcal{O})\in\mathbb{N}_0^p\times\mathcal{P}(\mathcal{I})\;\Big|\;
\yv_{\overline{\mathcal{O}},\mathcal{A}}
\text{ is flat for }\Sigma^{(\ell)}\Big\}.
\]
Among these, let $\mathcal{C}^{\mathcal{A},\mathbf{0}}_{(\Sigma,\yv)}\subset\mathcal{C}^{\mathcal{A}}_{(\Sigma,\yv)}$ be the subset whose validity \ set $\mathcal{B}_{\ell}(\xv^\circ;\Sigma^{(\ell)},\yv_{\overline{\mathcal{O}},\mathcal{A}})$ in the state of $\Sigma^{(\ell)}$ intersects the \emph{zero surface} of the removed channels and their derivatives:
\begin{equation}
\mathcal{B}_{\ell} \cap \big\{\xv_{\ell}:\ \uv_{\mathcal{A}}=\dot{\uv}_{\mathcal{A}}=\cdots=\uv_{\mathcal{A}}^{(\ell_{\mathcal{A}}-1)}=\mathbf 0\big\}\neq\emptyset.
\label{eq:validity_set_thm}
\end{equation}
\end{defn}

\begin{defn}[Zero-compatible flat-input complement family]\label{defn:F0}
We define
\[
\mathcal{F}^{\mathbf 0}_{(\Sigma,\yv)}
:=\Big\{\mathcal{A}\in\mathcal{F}_{(\Sigma,\yv)}\ \big|\ \mathcal{C}^{\mathcal{A},\mathbf 0}_{(\Sigma,\yv)}\neq\emptyset\Big\}.
\]
\end{defn}

\begin{thm}[Equivalence: dexterity $\Longleftrightarrow$ flat-input complement]\label{thm:dexterity-input}
Let $\Sigma$ be as in~\eqref{eq:sys} and let $\yv\in\R^p$ be a flat output. For any $\mathcal{A}\subset\mathcal{I}$,
\[
\mathcal{A}\in\mathcal{D}_{(\Sigma,\yv)}
\quad\Longleftrightarrow\quad
\mathcal{A}\in\mathcal{F}^{\mathbf 0}_{(\Sigma,\yv)}.
\]
\end{thm}

\begin{proof}
\emph{(ii $\Rightarrow$ i)} Assume $\mathcal{A}\in\mathcal{F}^{\mathbf 0}_{(\Sigma,\yv)}$. Then there exists $(\ell,\mathcal{O})\in\mathcal{C}^{\mathcal{A},\mathbf 0}_{(\Sigma,\yv)}$ such that $\yv_{\overline{\mathcal{O}},\mathcal{A}}$ is flat for $\Sigma^{(\ell)}$, with vector relative degree $\rv=\big[\rhov^\top\ \ell_{\mathcal{A}}^\top\big]^\top$ satisfying $|\rv|=n_{\ell}$, and with validity set $\mathcal{B}_{\ell}$ meeting the zero surface~\eqref{eq:validity_set_thm}. Ordering inputs so that $\uv_{\mathcal{A}}^{(\ell_{\mathcal{A}})}$ is last, the decoupling matrix has the block-triangular form
\[
\Am_{(\Sigma^{(\ell)},\,\yv^{(\rv)}_{\overline{\mathcal{O}},\mathcal{A}})}(\xv_{\ell})
=\begin{bmatrix}
\Am(\xv_{\ell}) & \Bm(\xv_{\ell})\\[2pt]
\mathbf 0 & \mathbf I_{|\mathcal{A}|}
\end{bmatrix}.
\]
Consider now $\Sigma^{(\ell)}_{\overline{\mathcal{A}}}$ with state dimension $n_{\ell,\overline{\mathcal{A}}}=n_{\ell}-|\ell_{\mathcal{A}}|$. The $\rhov$-th derivatives of $\yv_{\overline{\mathcal{O}}}$ yield the decoupling matrix
\(
\Am_{(\Sigma^{(\ell)}_{\overline{\mathcal{A}}},\,\yv^{(\rhov)}_{\overline{\mathcal{O}}})}(\overline{\xv})
=\Am(\xv_{\ell})\big|_{\uv_{\mathcal{A}}=\cdots=\uv_{\mathcal{A}}^{(\ell_{\mathcal{A}}-1)}=\mathbf 0}.
\)
By~\eqref{eq:validity_set_thm}, $\Am(\cdot)$ is nonsingular on that surface; moreover $|\rhov|=n_{\ell}-|\ell_{\mathcal{A}}|=n_{\ell,\overline{\mathcal{A}}}$. Hence $\yv_{\overline{\mathcal{O}}}$ is flat for $\Sigma^{(\ell)}_{\overline{\mathcal{A}}}$, so $\mathcal{A}\in\mathcal{D}_{(\Sigma,\yv)}$ by Definition~\ref{defn:dextsubset}.

\smallskip
\emph{(i $\Rightarrow$ ii)} Assume $\mathcal{A}\in\mathcal{D}_{(\Sigma,\yv)}$ and pick $(\widetilde{\ell},\mathcal{O})\in\overline{\mathcal{C}}^{\mathcal{A}}_{(\Sigma,\yv)}$ (Definition~\ref{defn:dextsubset}). Then $\yv_{\overline{\mathcal{O}}}$ is flat for $\Sigma^{(\widetilde{\ell})}_{\overline{\mathcal{A}}}$ with vector relative degree $\overline{\rhov}$ and nonsingular $\Am_{(\Sigma^{(\widetilde{\ell})}_{\overline{\mathcal{A}}},\,\yv^{(\overline{\rhov})}_{\overline{\mathcal{O}}})}(\overline{\xv})$. 
In the \emph{full} system, the derivatives $y_i^{(k)}$ ($i\in\overline{\mathcal{O}}$, $1\le k\le\overline{\rho}_i$) may depend on $u_j\in\mathcal{A}$; 
Let $c_i^j$ be the smallest order at which $u_j\in\mathcal{A}$ appears in some $y_i^{(k)}$, $1\le k\le\overline{\rho}_i$; if it never appears, set $c_i^j:=\overline{\rho}_i$. Define $s_i^j:=\overline{\rho}_i-c_i^j\ge 0$ and choose $l_{a_j}:=\max_{i\in\overline{\mathcal{O}}} s_i^j$.

Form $\ell:=\operatorname{merge}(\widetilde{\ell}_{\overline{\mathcal{A}}},\ell_{\mathcal{A}})$. 
Then $|\rv|=|\overline{\rhov}|+|\ell_{\mathcal{A}}|=n+|\widetilde{\ell}_{\overline{\mathcal{A}}}|+|\ell_{\mathcal{A}}|=n+|\ell|=n_\ell$, and the decoupling matrix
\[
\Am_{(\Sigma^{(\ell)},\,\yv^{(\rv)}_{\overline{\mathcal{O}},\mathcal{A}})}(\xv_{\ell})
=\begin{bmatrix}
\overline{\Am}(\xv_{\ell}) & \overline{\Bm}(\xv_{\ell})\\[2pt]
\mathbf 0 & \mathbf I
\end{bmatrix},
\]
is nonsingular on the zero surface since $\overline{\Am}(\xv_{\ell})\big|_{\uv_{\mathcal{A}}=\cdots=\uv_{\mathcal{A}}^{(\ell_{\mathcal{A}}-1)}=\mathbf 0}
=\Am_{(\Sigma^{(\widetilde{\ell})}_{\overline{\mathcal{A}}},\,\yv^{(\overline{\rhov})}_{\overline{\mathcal{O}}})}(\overline{\xv})$ is nonsingular. 
By continuity there exists a validity set $\mathcal{B}_{\ell}$ intersecting the zero surface, so $(\ell,\mathcal{O})\in\mathcal{C}^{\mathcal{A},\mathbf 0}_{(\Sigma,\yv)}$ and $\mathcal{A}\in\mathcal{F}^{\mathbf 0}_{(\Sigma,\yv)}$.
\end{proof}

\begin{rem}
(i) The proof does not impose a priori lower bounds on the prolongation entries $l_j$, and allows direct feedthrough (``relative degree zero'') of some $u_j$ into components of $\yv_{\overline{\mathcal{O}}}^{(\rhov)}$; the block-triangular structure still yields nonsingularity on the zero surface.   When $l_j=0$ for some channel, the corresponding input component does not appear in the output stack and no input-stack error filter is applied to it (cf. Notation).
(ii) In practice, the additional prolongations $\ell_{\mathcal A}$ can be chosen minimally by inspecting the first order where each $u_j$ enters the relevant derivatives of $y_i$-consistent with standard dynamic-extension constructions and recent pure-prolongation criteria~ \cite{Isidori1995,LevinePureProlongation2025}.
\end{rem}

\section{Dexterity Input-Output Negotiation Control}
\label{sec:negotiation-control}

\subsection{Why Transients Arise Under Direct Shutdown}
\label{subsec:transient-direct-shutdown}

Consider a system $\Sigma$ with a feedback-linearizing output $\yv$ and a dexterity subset of inputs indexed by $\mathcal{A}$ (Definition~\ref{defn:dextsubset}). 
Assume a feedback-linearizing controller $\mathcal{K}_\Sigma$ achieves exponential tracking of $\yv$.
Suppose that, at some time $t_s$, the inputs in $\mathcal{A}$ are intentionally deactivated (e.g., energy saving, cooling, actuator wear), and we \emph{switch} to a controller $\mathcal{K}_{\Sigma_{\overline{\mathcal{A}}}}$ designed for the reduced prolonged system $\Sigma^{(\ell)}_{\overline{\mathcal{A}}}$ using a realizing pair $(\ell,\mathcal{O})\in\overline{\mathcal{C}}^{\mathcal{A}}_{(\Sigma,\yv)}$.

Although this switch preserves exact linearizability of the \emph{reduced} output $\yv_{\overline{\mathcal{O}}}$, the closed-loop behavior of the \emph{kept} components need not persist \emph{without a transient}.
In particular, shutting inputs \emph{instantaneously} to zero may \emph{change} the vector relative degree of some kept outputs, introducing coupling with other states and yielding a non-negligible settling transient before exponential convergence is re-established.

\begin{exmp}[Transient after deactivating a single dexterity input]
\label{exmp:transient_u2_zero}
Reconsider the motivating example of Section~\ref{sec:motv_exmp} with $\yv=[x_1\ x_3\ x_4]^\top$. Since $\{2\}\in\mathcal{D}_{(\Sigma,\yv)}$, the input $u_2$ is a dexterity input and may be set to zero, while regulating $\yv_{\overline{\{3\}}}=[x_1\ x_3]^\top$ via $\Sigma_{\overline{\{2\}}}^{(\ell)}$. 
Figure~\ref{fig:sim_motv_exmpuu2_Zeros_transient} compares $\mathcal{K}_\Sigma$ (pre-switch) and $\mathcal{K}_{\Sigma_{\overline{\{2\}}}}$ (post-switch) for the reference $\yv^d=[x_1^d\ x_3^d\ x_4^d]^\top=[4\ 4\ -20]^\top$. 
At $t<8\si{s}$, the tracking error is exponentially stable. 
At $t\ge 8\si{s}$ we set $u_2\equiv 0$ and switch to $\mathcal{K}_{\Sigma_{\overline{\{2\}}}}$. 
While $x_1$ keeps converging smoothly to $x_1^d$, $x_3$ exhibits a transient immediately after the switch, because the relative degree of $y_2=x_3$ increases from $1$ to $2$ when $u_2$ is removed, coupling $x_3$ with $x_4$. 
For special references (e.g., $x_4^d=0$), the coupling may vanish and the transient can be negligible.
\end{exmp}

\begin{figure}[t]
    \centering
    \includegraphics[width=1.03\linewidth]{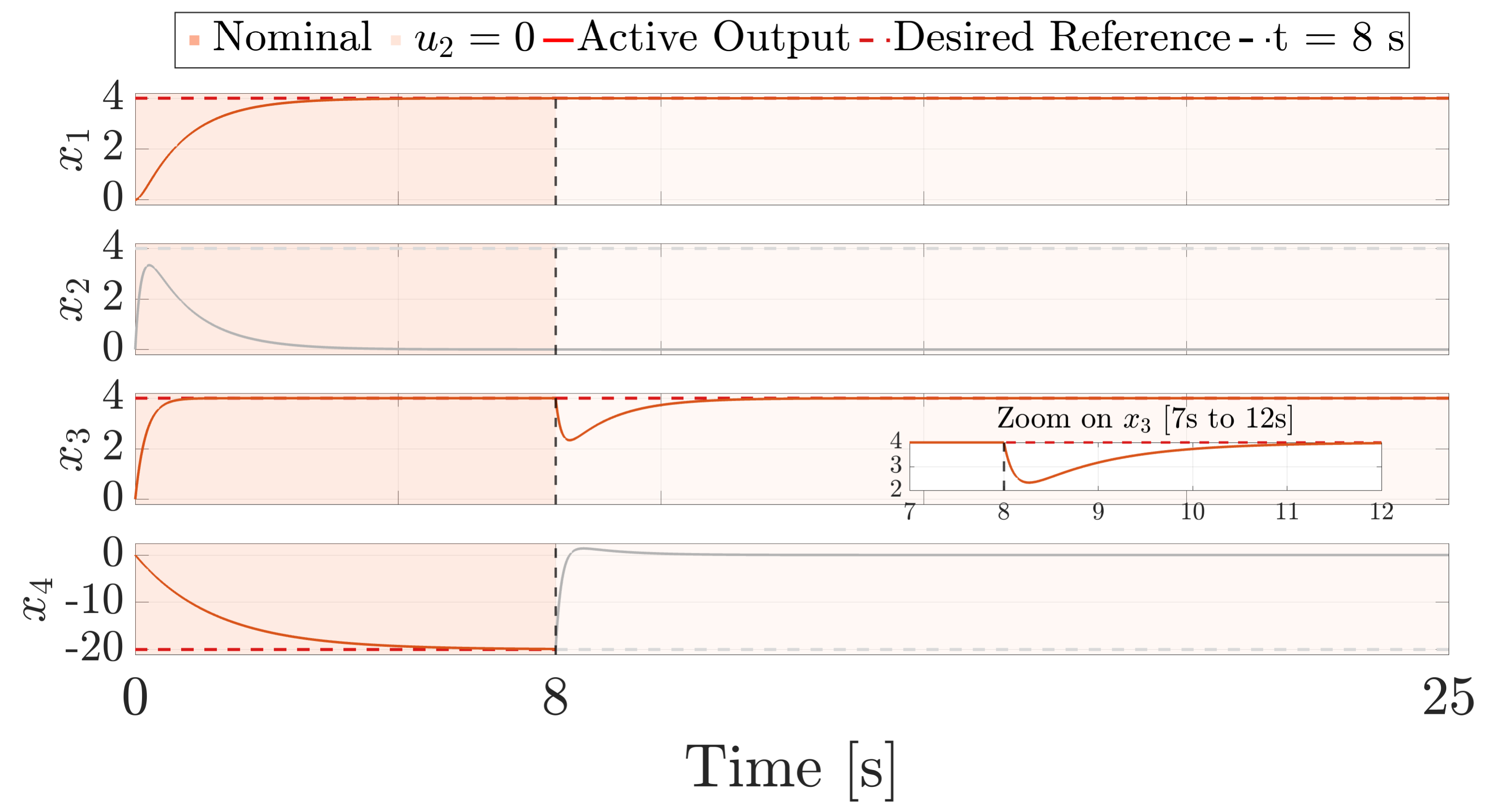}
    \caption{(Motivating example, Section~\ref{sec:motv_exmp}). 
    Direct shutdown: $\mathcal{K}_\Sigma$ regulates $\yv=[x_1\ x_3\ x_4]^\top$ until $t=8\si{s}$; then $u_2$ is set to zero and $\mathcal{K}_{\Sigma_{\overline{\{2\}}}}$ regulates $\yv_{\overline{\{3\}}}=[x_1\ x_3]^\top$. 
    A transient appears on $x_3$.}
    \label{fig:sim_motv_exmpuu2_Zeros_transient}
\end{figure}

\subsection{Eliminating the Transient via a Unified Prolonged Controller}
\label{subsec:no-transient-prolonged}

By Theorem~\ref{thm:dexterity-input} (dexterity $\Leftrightarrow$ flat-input complement with zero-surface compatibility), there exists a prolonged system $\Sigma^{(\ell)}$ and an index set $\mathcal{O}$ such that 
$\yv_{\overline{\mathcal{O}},\mathcal{A}}=[\yv_{\overline{\mathcal{O}}}^\top\ \ \uv_{\mathcal{A}}^\top]^\top$ is flat for $\Sigma^{(\ell)}$.
The \emph{additional} requirement is that $\yv$ itself be flat for the very same prolonged system $\Sigma^{(\ell)}$; only under this common prolongation can a single controller be used seamlessly across the switch.
In this setting the channels in $\mathcal{A}$ are treated as \emph{(prolonged) states} that can be \emph{driven to zero smoothly} by assigning virtual inputs on their integrator chains; switching between the two flat outputs $\yv$ and $\yv_{\overline{\mathcal{O}},\mathcal{A}}$ preserves the closed-loop dynamics of the shared components (meld switching result), hence \emph{no transient} appears on the kept outputs~ \cite{mizzoni2026switchingfeedbacklinearizingoutputsets}.

\begin{exmp}[No transient with a unified controller]
\label{exmp:no_transient}
Consider again the shutdown of $u_2$ in Example~\ref{exmp:transient_u2_zero}. 
Choose $\ell=\{0,1,0\}$ and use the \emph{same} linearizing controller $\mathcal{K}_{\Sigma^{(\ell)}}$ both before and after the shutdown.
At $t<8\si{s}$, regulate $\yv$; at $t\ge 8\si{s}$, regulate $\yv_{\overline{\{3\}},\{2\}}=[x_1\ x_3\ u_2]^\top$ and set the virtual command ${v}_2=-k^0_{u,2}u_2$ with $k^0_{u,2}>0$, yielding an exponential decay of $u_2$ with time constant $\tau=(k^0_{u,2})^{-1}$. 
Figure~\ref{fig:sim_motv_exmpuu2_Zeros} shows that, in contrast to Example~\ref{exmp:transient_u2_zero}, \emph{no transient} appears in the kept output $y_2=x_3$ at the switch.
\end{exmp}

\begin{figure}[t]
    \centering
    \includegraphics[width=1.03\linewidth]{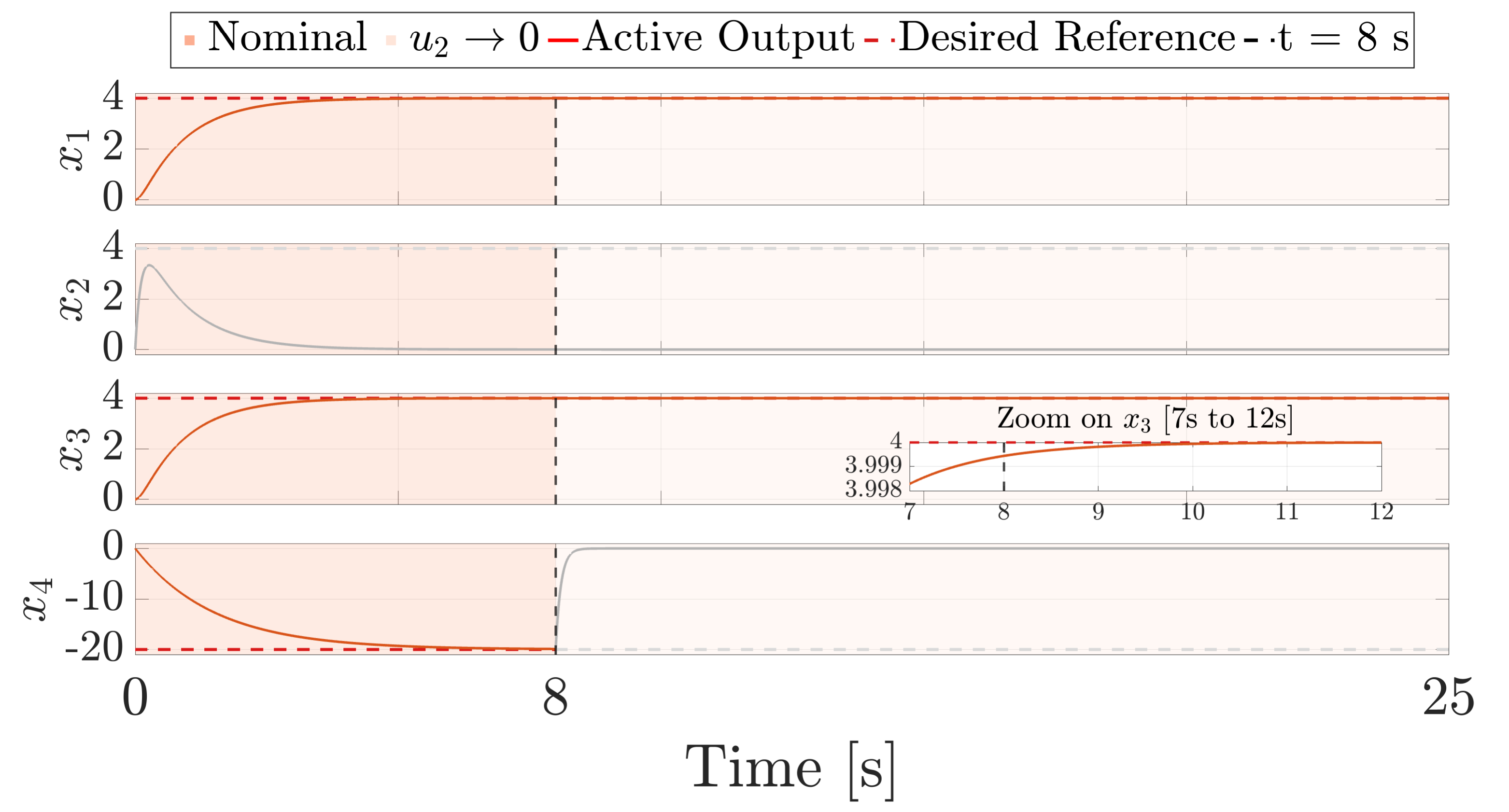}
    \caption{Unified prolonged control: one controller $\mathcal{K}_{\Sigma^{(\ell)}}$ regulates pre-switch $\yv$ and post-switch $\yv_{\overline{\{3\}},\{2\}}=[x_1\ x_3\ u_2]^\top$, and smoothly drives $u_2\to 0$ (no transient on the kept outputs).}
    \label{fig:sim_motv_exmpuu2_Zeros}
\end{figure}

\subsection{When a Unified Controller is Not Available}
\label{subsec:limit-unified}

The unified approach requires that the \emph{same} prolonged system $\Sigma^{(\ell)}$ admits \emph{both} $\yv$ and $\yv_{\overline{\mathcal{O}},\mathcal{A}}$ as flat outputs. 
This may fail: there exist dexterity subsets for which Theorem~\ref{thm:dexterity-input} ensures the existence of $(\ell,\mathcal{O})$ making $\yv_{\overline{\mathcal{O}},\mathcal{A}}$ flat, but the same $\Sigma^{(\ell)}$ does \emph{not} render $\yv$ flat.

\begin{exmp}[Unified controller not applicable]
\label{exmp:transient}
In the motivating example (Section~\ref{sec:motv_exmp}), since $\{1,2\}\in\mathcal{D}_{(\Sigma,\yv)}$, both $u_1$ and $u_2$ can be deactivated, leaving $\yv_{\overline{\{2,3\}}}=x_1$ as the controlled output via $\Sigma^{(\ell)}_{\overline{\{1,2\}}}$. 
Switching from $\mathcal{K}_\Sigma$ to $\mathcal{K}_{\Sigma_{\overline{\{1,2\}}}}$ (Fig.~\ref{fig:sim_motv_exmpu1u2_Zeros}) produces a transient on $x_1$ due to a jump in relative degree, despite $x_1$ being part of $\yv$. 
Theorem~\ref{thm:dexterity-input} guarantees that there exists $\ell$ (e.g., $\{2,2,0\}$) such that $\yv_{\overline{\mathcal{O}},\{1,2\}}$ is flat for $\Sigma^{(\ell)}$, but in this case $\yv$ is \emph{not} flat for the same $\Sigma^{(\ell)}$ (e.g., $\yv_{\overline{\{2,3\}}}$ has relative degree $4<n_\ell=8$). 
Therefore a unified controller cannot be used here, and transients may be unavoidable unless a different choice of task or prolongation is made.
\end{exmp}

\begin{figure}[t]
    \centering
    \includegraphics[width=1.03\linewidth]{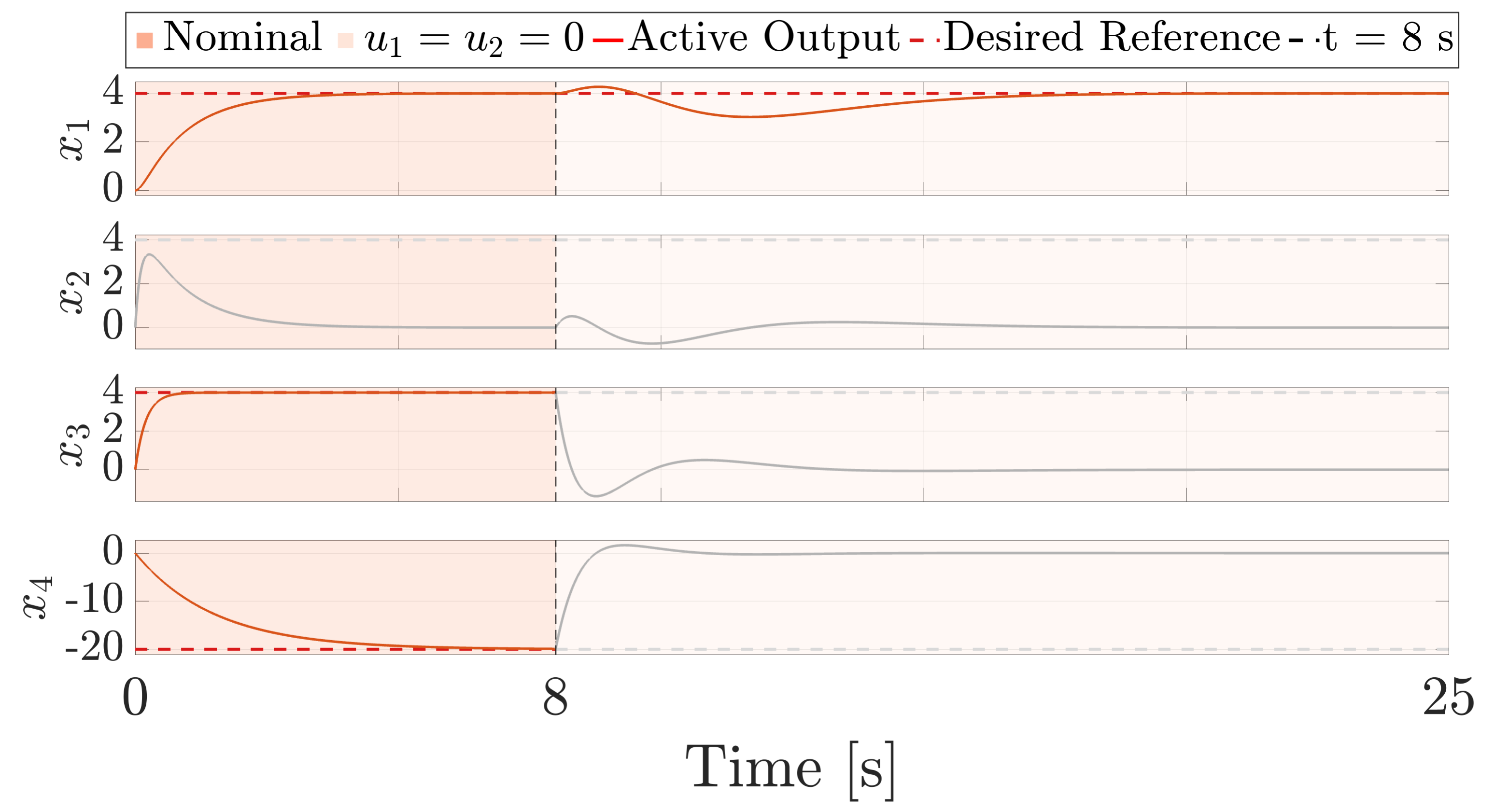}
    \caption{Direct shutdown of $u_1$ and $u_2$ at $t=8\si{s}$, then regulation of $\yv_{\overline{\{2,3\}}}=x_1$ with $\mathcal{K}_{\Sigma_{\overline{\{1,2\}}}}$. 
    A transient appears in $x_1$ due to a relative-degree jump.}
    \label{fig:sim_motv_exmpu1u2_Zeros}
\end{figure}

\medskip
The next subsections formalize (i) the conditions under which a \emph{unified} prolonged controller exists and guarantees zero-transient switching among compatible outputs, and (ii) a constructive design (gains, dwell-time, and switching logic) for negotiation among full and reduced tasks.

\section{Input-Output Subsets with a Common Prolongation}
\label{sec:common-prolongation}

In~ \cite{mizzoni2026switchingfeedbacklinearizingoutputsets} we introduced \emph{melds}, i.e., square selections of scalar outputs from a deck that yield a flat output for a given system.  We showed that if
the corresponding validity sets intersect on a nonempty compact set,  a feedback-linearizing controller can switch among them while preserving the exponentially stable error dynamics of the \emph{shared} components. 
In the present framework, three systems naturally arise: the original system $\Sigma$ with flat output $\yv$; the reduced prolonged system $\Sigma^{(\ell)}_{\overline{\mathcal{A}}}$ (with $(\ell,\mathcal{O})\in\overline{\mathcal{C}}^{\mathcal{A}}_{(\Sigma,\yv)}$) in which $\yv_{\overline{\mathcal{O}}}$ is flat; and the prolonged original system $\Sigma^{(\ell)}$ (with $(\ell,\mathcal{O})\in\mathcal{C}^{\mathcal{A}}_{(\Sigma,\yv)}$) in which $\yv_{\overline{\mathcal{O}},\mathcal{A}}$ is flat. At first glance, this multiplicity prevents a direct application of~ \cite{mizzoni2026switchingfeedbacklinearizingoutputsets}. However, if \emph{the original output $\yv$ is also flat for the same $\Sigma^{(\ell)}$}, the switching framework applies directly.

\subsection{Common-Prolongation Families}
\label{subsec:common-prolongation-families}

\subsubsection*{Prolongations that admit flat input-complements}
\begin{defn}
Let $\Sigma$ be as in~\eqref{eq:sys}, let $\yv\in\R^p$ be a flat output for $\Sigma$, and let $\mathcal{A}\in\mathcal{F}_{(\Sigma,\yv)}$ (Definition~\ref{defn:Ffamily}). 
Define the set of \emph{admissible prolongations} for $\mathcal{A}$ by
\[
\mathcal{L}^{\mathcal{A}}_{(\Sigma,\yv)}
:=\Big\{\ \ell\in\mathbb{N}_0^p\ \Big|\ \exists\,\mathcal{O}\in\mathcal{P}(\mathcal{I})\ \text{with}\ (\ell,\mathcal{O})\in\mathcal{C}^{\mathcal{A}}_{(\Sigma,\yv)}\ \Big\}.
\]
In particular, $\mathcal{L}^{\emptyset}_{(\Sigma,\yv)}$ collects all $\ell$ such that $\yv$ is a flat output for $\Sigma^{(\ell)}$.
\end{defn}

\subsubsection*{Pairs realizable under a fixed prolongation}
\begin{defn}\label{def:common_prolongation}
Let $\ell\in\mathcal{L}^{\emptyset}_{(\Sigma,\yv)}$. 
Define the \emph{$\ell$-realizable family}
\[
\mathcal{N}^{\ell}_{(\Sigma,\yv)}
:= \Big\{\,(\mathcal{A},\mathcal{O})\in \mathcal{F}_{(\Sigma,\yv)}\times\mathcal{P}(\mathcal{I})\ \Big|\ (\ell,\mathcal{O})\in\mathcal{C}^{\mathcal{A}}_{(\Sigma,\yv)} \Big\}.
\]
Each element induces a flat output $\yv_{\overline{\mathcal{O}},\mathcal{A}}=[\yv_{\overline{\mathcal{O}}}^\top\ \ \uv_{\mathcal{A}}^\top]^\top$ for the \emph{same} prolonged system $\Sigma^{(\ell)}$. 
By convention, $(\emptyset,\emptyset)\in\mathcal{N}^{\ell}_{(\Sigma,\yv)}$ (the original output $\yv$).
\end{defn}

\subsubsection*{Zero-transient switching under a common prolongation}
\begin{lem}[No-transient switch under a common prolongation on shared components]
\label{lem:zero_transient_common_ell}
Let $\ell\in\mathcal{L}^{\emptyset}_{(\Sigma,\yv)}$ and consider any two elements 
$v_i=(\mathcal{A}_i,\mathcal{O}_i)$ and $v_j=(\mathcal{A}_j,\mathcal{O}_j)$ in $\mathcal{N}^{\ell}_{(\Sigma,\yv)}$.
Assume the associated flat outputs $\yv_{\overline{\mathcal{O}_i},\mathcal{A}_i}$ and $\yv_{\overline{\mathcal{O}_j},\mathcal{A}_j}$ are \emph{compatible} (Section~\ref{subs:preliminaries}). 
Then under the \emph{same} exact-linearizing controller for $\Sigma^{(\ell)}$, a switch from regulating $\yv_{\overline{\mathcal{O}_i},\mathcal{A}_i}$ to regulating $\yv_{\overline{\mathcal{O}_j},\mathcal{A}_j}$ preserves the exponentially stable error dynamics of the \emph{shared} components, i.e., produces \emph{no transient} on outputs common to both selections.
\end{lem}
\begin{proof}[Sketch]
Both outputs are flat for the same $\Sigma^{(\ell)}$, and their validity sets overlap (compatibility). 
By~ \cite{mizzoni2026switchingfeedbacklinearizingoutputsets}, switching between such melds preserves the closed-loop error dynamics of shared components under the same linearizing feedback; here the input channels in $\mathcal{A}$ are part of the output in one selection and are smoothly driven (via $\ell$) in the other, which does not affect the shared-output dynamics. \qedhere
\end{proof}

\begin{rem}
A common prolongation $\ell$ is \emph{sufficient} to implement a single controller $\mathcal{K}_{\Sigma^{(\ell)}}$ across selections; however, if two outputs are not compatible (disjoint validity sets), switches attempted outside the intersection of validity sets (where the decoupling matrix for one selection is singular) are infeasible, regardless of $\ell$-this mirrors the standard requirement in exact linearization~ \cite{Isidori1995}.
\end{rem}

\subsection{Graph Representation and Negotiability}
\label{subsec:graph-negotiability}

\subsubsection*{Compatibility graph}
\begin{defn}[Compatibility graph under a common prolongation]
Let $\ell\in\mathcal{L}^{\emptyset}_{(\Sigma,\yv)}$. 
Define the undirected graph $\mathcal{G}^{\ell}=(\mathcal{N}^{\ell}_{(\Sigma,\yv)},\mathcal{E}^{\ell})$ with vertices $v=(\mathcal{A},\mathcal{O})\in\mathcal{N}^{\ell}_{(\Sigma,\yv)}$ and edges defined by
\[
(v_i,v_j)\in\mathcal{E}^{\ell}
\iff
\yv_{\overline{\mathcal{O}_i},\mathcal{A}_i}\ \text{and}\ \yv_{\overline{\mathcal{O}_j},\mathcal{A}_j}\ \text{are compatible on}\ \Sigma^{(\ell)}.
\]
\end{defn}

Let $\mathcal{G}^{\ell}_1,\ldots,\mathcal{G}^{\ell}_{k_\ell}$ be the connected components of $\mathcal{G}^{\ell}$ with vertex sets $\mathcal{N}^{\ell}_1,\ldots,\mathcal{N}^{\ell}_{k_\ell}$, respectively. 
Since $\ell\in\mathcal{L}^{\emptyset}_{(\Sigma,\yv)}$, the vertex $(\emptyset,\emptyset)$ is present and belongs to exactly one component.

\subsubsection*{Negotiable sets and their union}
\begin{defn}[Graph of $\ell$-negotiable outputs]
For each $\ell\in\mathcal{L}^{\emptyset}_{(\Sigma,\yv)}$, the \emph{graph of $\ell$-negotiable outputs} is the connected component of $\mathcal{G}^{\ell}$ that contains $(\emptyset,\emptyset)$:
\[
\mathcal{G}^{\ell}_{\star}=(\mathcal{N}^{\ell}_{\star,(\Sigma,\yv)},\mathcal{E}^{\ell}_{\star}).
\]
Its vertex set $\mathcal{N}^{\ell}_{\star,(\Sigma,\yv)}$ is the \emph{$\ell$-negotiable set}. 
If needed, enumerate its elements as 
\(
\mathcal{N}^{\ell}_{\star,(\Sigma,\yv)}=\{(\mathcal{A}^{\ell}_i,\mathcal{O}^{\ell}_i)\}_{i=1}^{d_\ell},
\)
with $d_\ell:=|\mathcal{N}^{\ell}_{\star,(\Sigma,\yv)}|$. 
Finally, define the overall \emph{negotiable set}
\[
\mathcal{N}_{\star,(\Sigma,\yv)} \;:=\; \bigcup_{\ell\in\mathcal{L}^{\emptyset}_{(\Sigma,\yv)}} \mathcal{N}^{\ell}_{\star,(\Sigma,\yv)}.
\]
\end{defn}

\begin{rem}[Interpretation]
\label{rem:negotiability}
The term ``$\ell$-negotiable'' reflects that both $\yv$ and any $\yv_{\overline{\mathcal{O}},\mathcal{A}}=[\yv_{\overline{\mathcal{O}}}^\top\ \ \uv_{\mathcal{A}}^\top]^\top$ with $\big((\emptyset,\emptyset),(\mathcal{A},\mathcal{O})\big)\in\mathcal{E}^{\ell}_{\star}$ are:
(i) flat outputs for the \emph{same} prolonged system $\Sigma^{(\ell)}$, and 
(ii) compatible (their validity sets intersect). 
By Lemma~\ref{lem:zero_transient_common_ell} and~ \cite{mizzoni2026switchingfeedbacklinearizingoutputsets}, one can ``negotiate'' between such outputs---i.e., swap components \(\uv_{\mathcal{A}}\) and \(\yv_{\mathcal{O}}\) in the linearizing output---without introducing transients on the shared components.
\end{rem}

\begin{table*}[t]
\centering
\caption{Summary of main symbols}
\label{tab:symbols}
\begin{tabular}{ll}
\hline
\textbf{Symbol} & \textbf{Meaning} \\
\hline
$\Sigma$ & Input-affine system $\dot \xv = \fv(\xv) + \sum_{i=1}^p \gv_i(\xv)u_i$ \\
$\xv \in \mathbb{R}^n$ & State \\
$\uv \in \mathbb{R}^p$ & Input vector, $\uv=[u_1\ \cdots\ u_p]^\top$ \\
$\yv=\hv(\xv)\in\mathbb{R}^p$ & Output (task), assumed flat (linearizing) on a validity set \\
$\mathcal{I}=\{1,\dots,p\}$ & Index set for input/output channels \\
$\mathcal{S}\subseteq \mathcal{I}$, $\overline{\mathcal{S}}$ & Subset of indices and its complement $\overline {\mathcal{S}} = \mathcal{I}\setminus \mathcal{S}$ \\
$\qv_{\mathcal{S}}$ & Subvector of $\qv\in \mathbb{R}^p$ with indices in $\mathcal{S}$ (preserving order) \\
$\mathrm{merge}_S(\qv',\qv'')$ & Merge operator: takes entries from $\qv'$ on $\mathcal{S}$ and from $\qv''$ on $\overline {\mathcal{S}}$ \\
$\rv=[r_1\ \cdots\ r_p]^\top$ & A (Vector) relative degree of $(\Sigma,\yv)$ \\
$|\rv|=\sum_i r_i$ & Sum of relative degrees \\
$\Am_{(\Sigma,\yv^{(\rv)})}(\xv)$ & Decoupling matrix appearing at $\yv^{(\rv)}$ of the system $\Sigma$\\
$\mathcal{B}(\xv^\circ;\Sigma,\yv)$ & Validity set of $\yv$ around $\xv^\circ$. It is defined as an open neighborhood  of $\xv^\circ$ where for all $\xv \in \mathcal{B}$, $|\Am(\xv)|\neq 0$. \\
$\Sigma_{\overline{\mathcal{A}}}$ & Reduced system keeping inputs in $\overline{\mathcal{A}}\subseteq \mathcal{I}$  \\
$\ell=(l_1,\dots,l_p)\in\mathbb{N}_0^p$ & Prolongation pattern (integrator-chain lengths per input) \\
$\mathbb{N}_0^p\big|_{\overline{\mathcal{A}}}$ & Patterns with $l_i=0$ for all $i\in {\mathcal{A}}$ (no chains on removed inputs) \\
$|\ell|=\sum_i l_i$ & Total prolongation order \\
$\Sigma^{(\ell)}$ & Prolonged system (dynamic extension) with virtual input $\vv=\uv^{(\ell)}$ \\
$\Sigma_{\overline{\mathcal{A}}}^{(\ell)}$ & System obtained by interconnecting $\Sigma_{\overline{\mathcal{A}}}$ with $p-|\mathcal{A}|$ chains of integrators of lengths $l_i$ on the surviving inputs\\
$\xv_\ell$ & Stacked prolonged state (includes $\xv$ and input-derivative stacks) \\
$\mathcal{A} \in \mathcal{D}_{(\Sigma,\yv)}$ & Dexterity subset: removing inputs in $A$ preserves flatness of a reduced task \\
$\delta^i_{(\Sigma,\yv)}\in\{1,\ldots,p-1\}$ & Minimum dexterity loss for input $u_i$ (min $|\mathcal{A}|$ over dexterity sets containing $i$) \\
$\mathcal{A}\in \mathcal{F}_{(\Sigma,\yv)}$ & Flat-input-complement subset (inputs in $\mathcal{A}$ can enter the output under prolongation) \\
$\mathcal{F}^{\mathbf{0}}_{(\Sigma,\yv)}$ & Zero-compatible flat-input-complement subsets (validity set intersects zero surface) \\
$\mathcal{C}^{\mathcal{A}}_{(\Sigma,\yv)}$ & Realizing pairs $(\ell,\mathcal{O})$ making the required output flat (for dexterity or complement) \\
$\mathcal{L}^{\mathcal{A}}_{(\Sigma,\yv)}$ & Admissible prolongations $\ell$ for a given $\mathcal{A}$ \\
$\mathcal{N}^\ell_{(\Sigma,\yv)}$ & $\ell$-realizable pairs $(\mathcal{A},\mathcal{O})$ yielding flat outputs $[\yv_{\overline{\mathcal{O}}}^\top\ \uv_{\mathcal{A}}^\top]^\top$ for the prolonged system $\Sigma^\ell$\\
$\mathcal{G}^\ell=(\mathcal{N}^\ell,\mathcal{E}^\ell)$ & Compatibility graph under fixed $\ell$ \\
$\mathcal{G}^\ell_\star$ & Connected component of $\mathcal{G}_\ell$ containing $(\varnothing,\varnothing)$ (negotiable set) \\
$\eta(t)$ & Switching signal selecting a vertex (i.e., a pair $(\mathcal{A},\mathcal{O})$) \\
$\yv_{\eta(t)}$ & Active output vector selected by the switching signal $\eta$\\
$\Gammam_{\eta(t)}$ & Selection matrix used to build the active output $\yv_{\eta(t)}$ \\
$\wv\in\mathbb{R}^p$ & Stacked reference + error-feedback term for the virtual input \\
$\vv\in\mathbb{R}^p$ & Virtual input of prolonged system ($\vv=\uv^{(\ell)}$) \\
\hline
\end{tabular}
\end{table*}

\subsection{Summary}
\label{subsec:summary-sets}

Figure~\ref{fig:setDrawing} summarizes the sets introduced throughout the section. For ease of reference, we list the core definitions:
\[
\small
\begin{aligned}
\mathcal{F}_{(\Sigma,\yv)}
&=\big\{\mathcal{A}\subset\mathcal{I}\ \big|\ \mathcal{A}\ \text{ flat-input complement subset for }(\Sigma,\yv)\big\},\\
\mathcal{L}^{\mathcal{A}}_{(\Sigma,\yv)}
&= \Big\{\ \ell\in\mathbb{N}_0^{p}\ :\ \exists\,\mathcal{O}\in\mathcal{P}(\mathcal{I})\ \text{s.t.}\ (\ell,\mathcal{O})\in\mathcal{C}^{\mathcal{A}}_{(\Sigma,\yv)}\ \Big\},\\[2pt]
\Omega^{\mathcal{A}}_{(\Sigma,\yv)}
&= \Big\{\ \mathcal{O}\in\mathcal{P}(\mathcal{I})\ :\ \exists\,\ell\in\mathbb{N}_0^{p}\ \text{s.t.}\ (\ell,\mathcal{O})\in\mathcal{C}^{\mathcal{A}}_{(\Sigma,\yv)}\ \Big\},\\[2pt]
\mathcal{N}^{\ell}_{(\Sigma,\yv)}
&= \Big\{\ (\mathcal{A},\mathcal{O})\in\mathcal{F}_{(\Sigma,\yv)}\times\mathcal{P}(\mathcal{I})\ :\ (\ell,\mathcal{O})\in\mathcal{C}^{\mathcal{A}}_{(\Sigma,\yv)}\ \Big\},\\[2pt]
\mathcal{C}^{\mathcal{A}}_{(\Sigma,\yv)}
&= \Big\{\,(\ell,\mathcal{O})\in\mathbb{N}_0^{p}\times\mathcal{P}(\mathcal{I})\ \Big|\ 
\yv_{\overline{\mathcal{O}},\mathcal{A}} \ \text{is a flat output for}\ \Sigma^{(\ell)}\ \Big\}.
\end{aligned}
\]
\begin{figure}[t!]
    \centering
\includegraphics[width=0.90\linewidth]{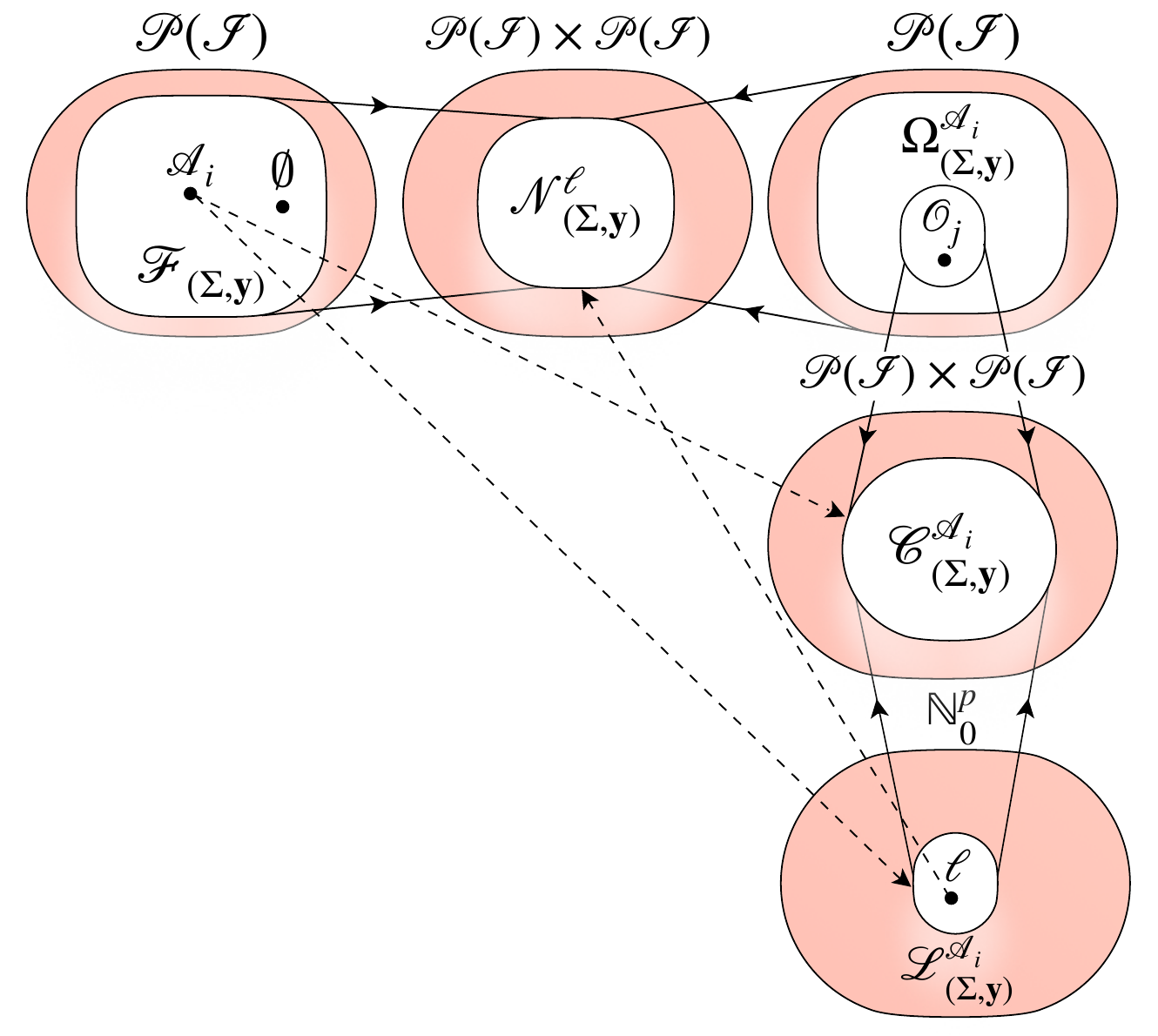}
    \caption{Schematic of the families introduced in this work for $(\Sigma,\yv)$:  flat-input complement subsets  $
\mathcal{F}_{(\Sigma,\yv)};$  admissible prolongations $\mathcal{L}^{\mathcal{A}}_{(\Sigma,\yv)}$; admissible index sets $\Omega^{\mathcal{A}}_{(\Sigma,\yv)}$; $\ell$-realizable family $\mathcal{N}^{\ell}_{(\Sigma,\yv)}$; and realizing pairs $\mathcal{C}^{\mathcal{A}}_{(\Sigma,\yv)}$.}
    \label{fig:setDrawing}
\end{figure}

\section{Switching Control Between Negotiable Pairs with a Common Prolongation}
\label{sec:switching-negotiable-common-ell}

We now provide a control scheme that enables switching between the original flat output $\yv$ and any adjacent flat output indexed by a vertex in the $\ell$-negotiable set $\mathcal{N}^{\ell}_{\star,(\Sigma,\yv)}$ (Definition~\ref{def:common_prolongation}), while inducing \emph{no undesired transient} on the shared output components. This addresses the question posed at the start of Section~\ref{sec:negotiation-control}.

\subsubsection*{Unified exact-linearizing structure on $\Sigma^{(\ell)}$}

Let  $\xv_\ell^\circ$ an operating point. Fix any $\ell\in\mathcal{L}^{\emptyset}_{(\Sigma,\yv)}$. By definition, $\yv$ is a flat output for $\Sigma^{(\ell)}$, hence there exists a validity neighborhood $\mathcal{B}_{\ell}$ of $\xv^\circ_\ell$ where
\[
\yv^{(\rv)} \;=\; \bv(\xv_\ell)\;+\;\Am_{(\Sigma^{(\ell)},\,\yv^{(\rv)})}(\xv_\ell)\,\vv,
\qquad |\rv|=n+|\ell|,\ \ \rank \Am = p.
\]
Introduce the stacked map and block matrix
\[
\qv(\xv_\ell) \;:=\; \begin{bmatrix}\bv(\xv_\ell)\\[2pt]\mathbf{0}_{p}\end{bmatrix}\in\R^{2p},\;\;
\Dm(\xv_\ell) \;:=\; 
\begin{bmatrix}
\Am_{(\Sigma^{(\ell)},\,\yv^{(\rv)})}(\xv_\ell)\\[2pt]
\mathbf{I}_{p}
\end{bmatrix}\in\R^{2p\times p}.
\]
Thus, for any selection of $p$ scalar channels composed of $\yv$ and $\uv$ (with appropriate derivative orders), the corresponding linearizing relation is obtained by \emph{selecting rows} of $\qv$ and $\Dm$.

\subsubsection*{Selection along the negotiable graph}
Let $\eta:(0,\infty)\to\{1,\ldots,d_\ell\}$ be a switching signal taking values on the vertex set of $\mathcal{G}^{\ell}_{\star}$, where the vertex $v_{\eta(t)}$ corresponds to the pair $(\mathcal{A}^{\ell}_{\eta(t)},\mathcal{O}^{\ell}_{\eta(t)})$. Define the selection matrix $\Gammam_{\eta(t)}\in\R^{p\times 2p}$ as
\[
\Gammam_{\eta(t)} \;=\;
\begin{bmatrix}
\Gammam^{\yv}_{\eta(t)} & \mathbf{0}\\
\mathbf{0} & \Gammam^{\uv}_{\eta(t)}
\end{bmatrix},
\]
\[
\Gammam^{\yv}_{\eta(t)}\in\R^{(p-|\mathcal{O}^{\ell}_{\eta(t)}|)\times p},\quad
\Gammam^{\uv}_{\eta(t)}\in\R^{|\mathcal{O}^{\ell}_{\eta(t)}|\times p},
\]
where
\[
\Gammam^{\yv}_{\eta(t)} \;=\; \big[\;\cdots\ \ev_{j}\ \cdots\big]^\top_{j\in\overline{\mathcal{O}^{\ell}_{\eta(t)}}},
\qquad
\Gammam^{\uv}_{\eta(t)} \;=\; \big[\;\cdots\ \ev_{j}\ \cdots\big]^\top_{j\in\mathcal{A}^{\ell}_{\eta(t)}}.
\]
Here $\ev_j$ is the $j$-th canonical vector of $\R^{p}$. The selected output is then
\begin{equation}
\yv_{\eta(t)} \;:=\; \Gammam_{\eta(t)}\!\begin{bmatrix}\yv\\ \uv\end{bmatrix}
\quad\in\R^{p},
\end{equation}
and its derivative orders are compactly encoded as
\begin{equation}
\rv_{\eta(t)} \;:=\; \Gammam_{\eta(t)}\,\begin{bmatrix}\rv\\ \ell\end{bmatrix}\ \in\ \mathbb{N}_0^{p}.
\end{equation}
With these definitions, on $\mathcal{B}_{\ell}$ we have the linear decoupled form
\begin{equation}
\label{eq:selected-io}
\yv_{\eta(t)}^{(\rv_{\eta(t)})} \;=\; \Gammam_{\eta(t)}\,\qv(\xv_\ell)\;+\;\Gammam_{\eta(t)}\,\Dm(\xv_\ell)\,\vv,
\end{equation}
and $\Gammam_{\eta(t)}\Dm(\xv_\ell)\in\R^{p\times p}$ is nonsingular (it is the decoupling matrix for the currently active selection).

\subsubsection*{Reference shaping and virtual inputs}
For each channel $i$, define the tracking vectors (or commonly known as $r_i-1$ ($l_i-1$)-\emph{jets},~\cite{Saunders1989}):
\[
\overline{\yv}_i:=\big[y_i,\dot y_i,\ldots,y_i^{(r_i-1)}\big]^{\!\top},\qquad
\overline{\uv}_i:=\big[u_i,\dot u_i,\ldots,u_i^{(l_i-1)}\big]^{\!\top},
\]
and reference stacks $\overline{\yv}^d_i$ and $\overline{\uv}^d_i$ built similarly from $y_i^d(\cdot)$ and $u_i^d(\cdot)$. Let the gains
\(
\kv_{y,i}^\top=[k_{y,i}^0\,\cdots\,k_{y,i}^{(r_i-1)}]
\)
and
\(
\kv_{u,i}^\top=[k_{u,i}^0\,\cdots\,k_{u,i}^{(l_i-1)}]
\)
be chosen so that the polynomials 
\(
\lambda^{r_i}+k_{y,i}^{(r_i-1)}\lambda^{r_i-1}+\cdots+k_{y,i}^0
\)
and 
\(
\lambda^{l_i}+k_{u,i}^{(l_i-1)}\lambda^{l_i-1}+\cdots+k_{u,i}^0
\)
are Hurwitz. Define the (time-varying) $2p$-vector
\begin{equation}
\label{eq:w}
\wv\;:=\;
\begin{bmatrix}\yv^{d,(\rv)}\\ \uv^{d,(\ell)}\end{bmatrix}
\;+\;
\sum_{i=1}^{p}\kv_{y,i}^\top(\overline{\yv}^d_i-\overline{\yv}_i)\,\ev_i
\;+\;
\sum_{i=1}^{p}\kv_{u,i}^\top(\overline{\uv}^d_i-\overline{\uv}_i)\,\ev_{p+i},
\end{equation}
where now $\ev_k$ denotes the $k$-th canonical vector of $\R^{2p}$.

\subsubsection*{Switching law}
At each time $t$, apply the linearizing input
\begin{equation}
\label{eq:switching-law}
\uv \;=\; \big(\Gammam_{\eta(t)}\Dm(\xv_\ell)\big)^{-1}\,\big[-\Gammam_{\eta(t)}\qv(\xv_\ell)+\Gammam_{\eta(t)}\wv\big].
\end{equation}
on the prolonged system $\Sigma^{(\ell)}$. By \eqref{eq:selected-io}--\eqref{eq:switching-law}, the closed-loop error dynamics on each active channel are linear time-invariant with preassigned characteristic polynomials, independent of which indices come from $\yv$ or from $\uv$.

\subsubsection*{Main guarantee}
\begin{prop}
Consider $\Sigma$ and fix $\ell\in\mathcal{L}^{\emptyset}_{(\Sigma,\yv)}$. Let $\eta(\cdot)$ take values on the vertex set of the graph $\mathcal{G}^{\ell}_{\star}$ (Section~\ref{subsec:graph-negotiability}). Assume that every \emph{consecutive} switch is along an edge of $\mathcal{G}^{\ell}_{\star}$, i.e.,
\(
\big(v_{\eta(t_k)},\,v_{\eta(t_{k+1})}\big)\in\mathcal{E}^{\ell}_{\star}
\qquad \text{for all consecutive switching instants } t_k<t_{k+1}.
\)
Then the closed-loop prolonged system under \eqref{eq:switching-law} satisfies:
\begin{enumerate}
\item For any interval $[t_k,t_{k+1})$, the error dynamics are exponentially stable for all channels indexed by
\(
\bigcap_{t\in[t_k,t_{k+1})}\overline{\mathcal{O}}^{\ell}_{\eta(t)}
\)
and
\(
\bigcap_{t\in[t_k,t_{k+1})}\mathcal{A}^{\ell}_{\eta(t)}
\)
(i.e., no transient on the outputs shared across the switch).
\item If the desired trajectories are bounded and the switching dwell-time satisfies the standard condition in~ \cite{mizzoni2026switchingfeedbacklinearizingoutputsets}, then the closed-loop state $\xv_\ell$ is uniformly bounded.
\end{enumerate}
\end{prop}

\begin{proof}
Item 1 follows from Lemma~\ref{lem:zero_transient_common_ell}: consecutive vertices joined by an edge correspond to compatible outputs for the \emph{same} prolonged system $\Sigma^{(\ell)}$, hence shared components preserve their exponentially stable dynamics under the common linearizing law \eqref{eq:switching-law}. Item 2 is a direct application of the state-boundedness result (with dwell-time) in~ \cite{mizzoni2026switchingfeedbacklinearizingoutputsets}.
\end{proof}

A summary of the main symbols introduced in this work is provided in Table~\ref{tab:symbols}.

\section{Application to Rigid-Body Flight Dynamics}\label{sec:rigid-body}

We illustrate the framework on a mechanically relevant use case with direct impact in robotics and aerospace. 
Consider the fully actuated flying platform $\Sigma$ (six independent force/torque channels in the body frame), whose dynamics (locally and away from Euler singularities) can be written as
\begin{equation}
\begin{aligned}
\Sigma:\;
\begin{bmatrix}
\dot{\pv}\\ \dot{\vv}\\ \dot{\Phim}\\ \dot{\Omegav}
\end{bmatrix}
&=
\underbrace{
\begin{bmatrix}
\vv\\
-\,g\,\ev_3\\
\Tm(\Phim)\,\Omegav\\
-\,\Jm^{-1}\big(\Omegav\times \Jm\Omegav\big)
\end{bmatrix}
}_{\fv(\xv)}
\;+\;
\underbrace{
\begin{bmatrix}
\mathbf 0 & \mathbf 0\\[2pt]
\frac{1}{m}\,\Rm(\Phim) & \mathbf 0\\[2pt]
\mathbf 0 & \mathbf 0\\[2pt]
\mathbf 0 & \Jm^{-1}
\end{bmatrix}
}_{[\gv_1\ \cdots\ \gv_6]}
\begin{bmatrix}\fv\\ \tauv\end{bmatrix},\\[2pt]
\yv &= \begin{bmatrix}\pv^\top & \Phim^\top\end{bmatrix}^\top,
\end{aligned}
\label{eq:sigma}
\end{equation}
where $\pv,\vv\in\mathbb{R}^3$ are position and linear velocity in the inertial frame, 
$\Phim=[\phi\ \theta\ \psi]^\top$ are roll-pitch-yaw angles, 
$\Omegav\in\mathbb{R}^3$ is the body angular velocity, $m>0$ is the mass, $\Jm\in\mathbb{R}^{3\times 3}$ is constant, symmetric, and positive definite, $g>0$ is gravity, $\ev_3=[0\ 0\ 1]^\top$, and $\Rm(\Phim)\in \SO$ is the body-to-inertial rotation. 
The kinematic map $\Tm(\Phim)$ is the standard Euler transformation (invertible for $\theta\neq \pm \tfrac{\pi}{2}$). 
We collect the inputs as
\(
\uv=\begin{bmatrix}u_1&\cdots&u_6\end{bmatrix}^\top
=\begin{bmatrix}\fv^\top & \tauv^\top\end{bmatrix}^\top.
\)
The state dimension is $n=12$.

\subsection{Characterization of the Inputs}\label{subsec:rb-character}
For~\eqref{eq:sigma}, the pose output $\yv=[\pv^\top\ \Phim^\top]^\top$ has vector relative degree 
$r_{\pv}=2$ w.r.t.\ force channels and $r_{\Phim}=2$ w.r.t.\ torque channels, hence $|\rv|=12=n$ and $\yv$ is a flat output on a nonempty validity set (Section~\ref{subs:preliminaries}). 
We now classify the inputs relative to this task.

\begin{prop}\label{prop:dex_thrust}
The three force channels $u_1,u_2,u_3$ are \emph{dexterity inputs} for $(\Sigma,\yv)$ with minimum loss $\delta^1_{(\Sigma,\yv)}=\delta^2_{(\Sigma,\yv)}=\delta^3_{(\Sigma,\yv)}=1$.
\end{prop}
\begin{proof}[Proof sketch]
By Theorem~\ref{thm:dexterity-input}, it suffices to exhibit a prolongation $\ell$ and an index set $\mathcal{O}$ such that, on $\Sigma^{(\ell)}$, the augmented output $\yv_{\overline{\mathcal{O}},\mathcal{A}}$ is flat and its validity set intersects the zero surface of the removed channels. 
Choose $\ell=(2,2,2,0,0,0)$ (two integrators on the force channels) and let $\mathcal{A}=\{i\}$ for $i\in\{1,2,3\}$. 
Table~\ref{tab:melds_fly_robot_compact} lists multiple choices of $\mathcal{O}$ producing flat outputs on $\Sigma^{(\ell)}$ with $(\ell,\mathcal{O})\in\mathcal{C}^{\mathcal{A},\mathbf 0}_{(\Sigma,\yv)}$, which, by the equivalence, proves that $\{i\}$ is a dexterity subset. Minimality (loss one) follows since singletons already satisfy the definition.
\end{proof}

\subsubsection*{Melds and Negotiation Under a Common Prolongation}
Fix $\ell=(2,2,2,0,0,0)$, for which the original output $\yv$ remains flat on $\Sigma^{(\ell)}$, i.e., $\ell\in\mathcal{L}^{\emptyset}_{(\Sigma,\yv)}$ (Section~\ref{subsec:common-prolongation-families}). 
All admissible flat outputs for $\Sigma^{(\ell)}$ built from $\yv$ and (a subset of) the force channels are reported in Table~\ref{tab:melds_fly_robot_compact}, together with the corresponding dexterity subset $\mathcal{A}$. 
As expected, no dexterity subset contains torque inputs (their removal destroys exact linearizability of any nontrivial subset of $\yv$), whereas several valid selections exist when removing one or two force channels. 
Note that some combinations are \emph{not} melds (e.g., $\{\pv,\psi,\theta,f_1\}$) because the associated decoupling matrix is singular---an artifact of the Euler parametrization; the \textit{exclusions} in the rightmost column encode validity-set conditions under which nonsingularity is preserved.
For instance, an existing platform where the flat output is $\{\pv,\theta,\psi\}$ (i.e., in our setting equivalent to $\#4$) is shown in~\cite{11373888}.

\begin{table}[b]

\caption{Flat outputs for $\Sigma^{(\ell)}$, $\ell=(2,2,2,0,0,0)$ (intrinsic representation) }
\scriptsize
\setlength{\tabcolsep}{2pt}
\label{tab:melds_fly_robot_compact}

\begin{tabularx}{1\linewidth}{@{}c l c c c c c@{}}
\toprule
\# & \makecell{\textbf{Chosen outputs}} & $\mathcal{A}$ & $\mathcal{O}$&
\makecell{\textbf{Mode}} & $p-|\mathcal{A|}$ & \makecell{$\overline{\mathcal{B}}$}\\
\midrule
$1$  & $\pv,\phi,\theta,\psi$ & $\emptyset$ &  $\emptyset$& FM & 6 & $\emptyset$\\
$2$  & $\pv,\phi,\psi,f_1$    & $\{1\}$  &  $\{5\}$    & DF & 5 & $\Exc{-f_3\cos\phi - f_2\sin\phi}$ \\
$3$  & $\pv,\phi,\theta,f_1$  & $\{1\}$  & \{6\}   & DF & 5 & $\Exc{f_2\cos\phi - f_3\sin\phi}$ \\
$4$  & $\pv,\psi,\theta,f_2$  & $\{2\}$ & \{4\}    & DF & 5 & $\Exc{f_3}$ \\
$5$  & $\pv,\phi,\theta,f_2$  & $\{2\}$  & $\{6\}$   & DF & 5 & $\Exc{f_3\sin\theta + f_1\cos\phi\cos\theta}$ \\
$6$  & $\pv,\phi,\psi,f_2$    & $\{2\}$ & $\{5\}$     & DF & 5 & $\Exc{f_1\sin\phi}$ \\
$7$  & $\pv,\phi,\psi,,f_3$    & $\{3\}$ & $\{5\}$      & DF & 5 & $\Exc{f_1\cos\phi}$ \\
$8$  & $\pv,\theta,\psi,f_3$  & $\{3\}$  & $\{4\}$     & DF & 5 & $\Exc{f_2}$ \\
$9$  & $\pv,\phi,\theta,f_3$  & $\{3\}$   & $\{6\}$   & DF & 5 &  $\Exc{f_2\sin\theta + f_1\cos\theta\sin\phi}$ \\
${10}$ & $\pv,\theta,f_{1,2}$ & $\{1,2\}$  & $\{4,6\}$ & QM & 4 & $\Exc{-f_3^2\sin\phi - f_2 f_3\cos\phi}$ \\
${11}$ & $\pv,\theta,f_{2,3}$ & $\{2,3\}$  & $\{4,6\}$ & QM & 4 & $\Exc{f_1(f_2\cos\phi - f_3\sin\phi)}$ \\
${12}$ & $\pv,\theta,f_{1,3}$ & $\{1,3\}$& $\{4,6\}$ &  QM & 4 & $\Exc{f_2(f_2\cos\phi - f_3\sin\phi)}$\\
${13}$ & $\pv,\psi,f_{1,2}$   & $\{1,2\}$ & $\{4,5\}$ & QM & 4 & $\Exc{f_3(f_3\cos\phi + f_2\sin\phi)}$ \\
${14}$ & $\pv,\psi,f_{1,3}$   & $\{1,3\}$ & $\{4,5\}$& QM & 4 & $\Exc{-f_2(f_3\cos\phi + f_2\sin\phi)}$ \\
${15}$ & $\pv,\psi,f_{2,3}$   & $\{2,3\}$ & $\{4,5\}$& QM & 4 & $\Exc{f_1(f_3\cos\phi + f_2\sin\phi)}$ \\
${16}$ & $\pv,\phi,f_{1,2}$   & $\{1,2\}$& $\{5,6\}$ & QM & 4 & C16 \\
${17}$ & $\pv,\phi,f_{1,3}$   & $\{1,3\}$ & $\{5,6\}$& QM & 4 & C17 \\
${18}$ & $\pv,\phi,f_{2,3}$   & $\{2,3\}$ & $\{5,6\}$& QM & 4 & C18 \\
\bottomrule
\end{tabularx}
\vspace{2pt}
\raggedright\scriptsize
\textbf{Legend:} FM = Full Act. Mode, DF = Dual-Force, QM = Quad-Mode.\newline 

C16: $\Exc{f_3(f_1\cos\theta + f_3\cos\phi\sin\theta + f_2\sin\phi\sin\theta)}$;\,\newline
C17: $\Exc{f_2(f_1\cos\theta + f_3\cos\phi\sin\theta + f_2\sin\phi\sin\theta)}$;\,\newline
C18: $\Exc{f_1(f_1\cos\theta + f_3\cos\phi\sin\theta + f_2\sin\phi\sin\theta)}$.
\end{table}

\subsubsection*{Numerical Simulation}\label{subsec:rb-sim}

\arxivtext{%
\begin{figure}[t]
   \centering
   \includegraphics[width=1\linewidth,trim=0.2cm 0cm 0cm 0cm, clip]{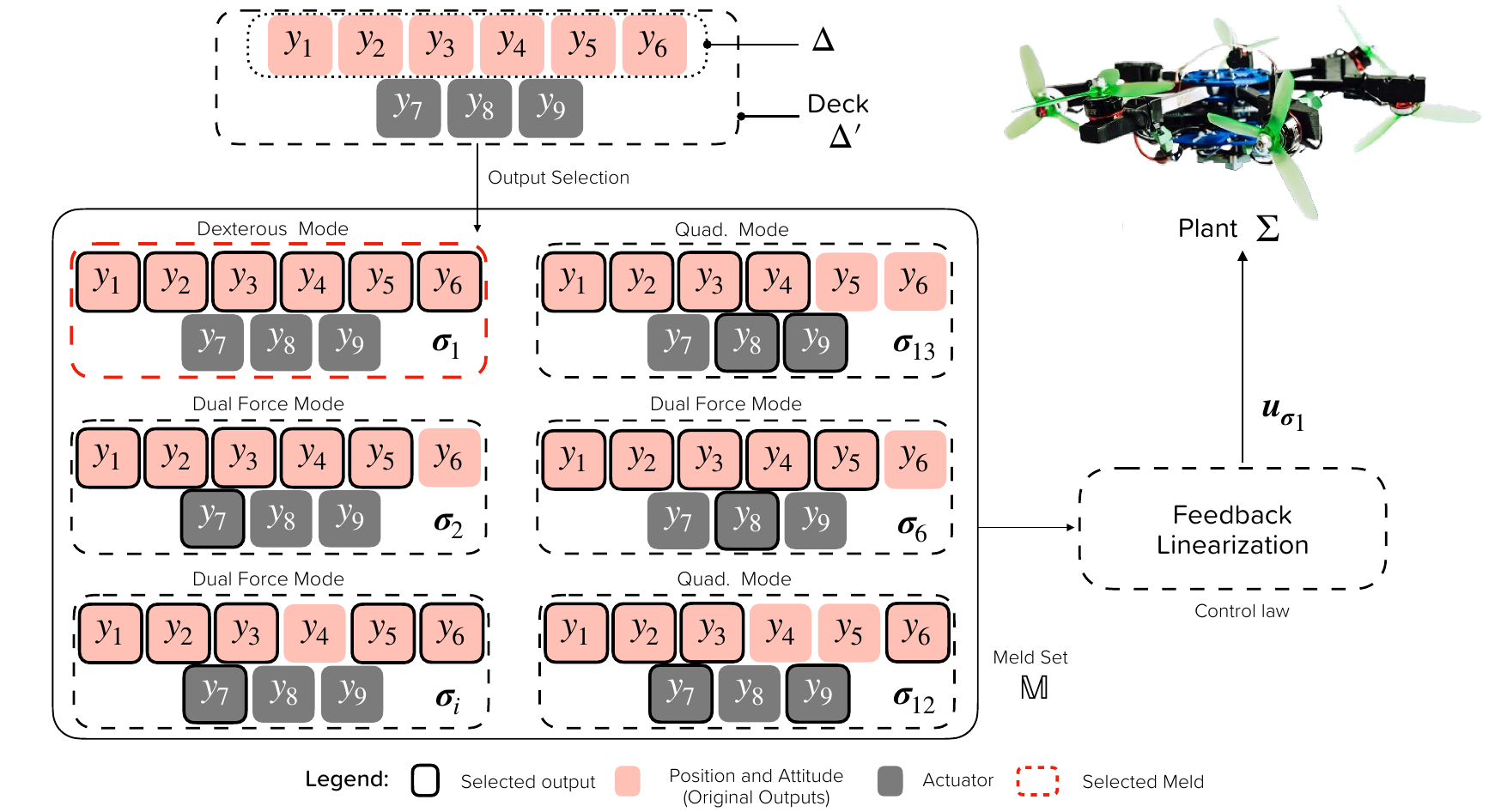}
   \caption{Deck $\Delta^{\ell}$ for $\ell=(2,2,2,0,0,0)$: pose variables plus the three dexterity thrust forces. Nodes are admissible melds (a subset is shown), including the full-pose meld \#1; edges denote compatibility (overlapping validity sets).}
   \label{fig:meld_picture}
\end{figure}}

\begin{figure}[t]
    \centering
    \includegraphics[trim =0.61cm 0cm 0cm 0cm, clip, width=1.062\linewidth]{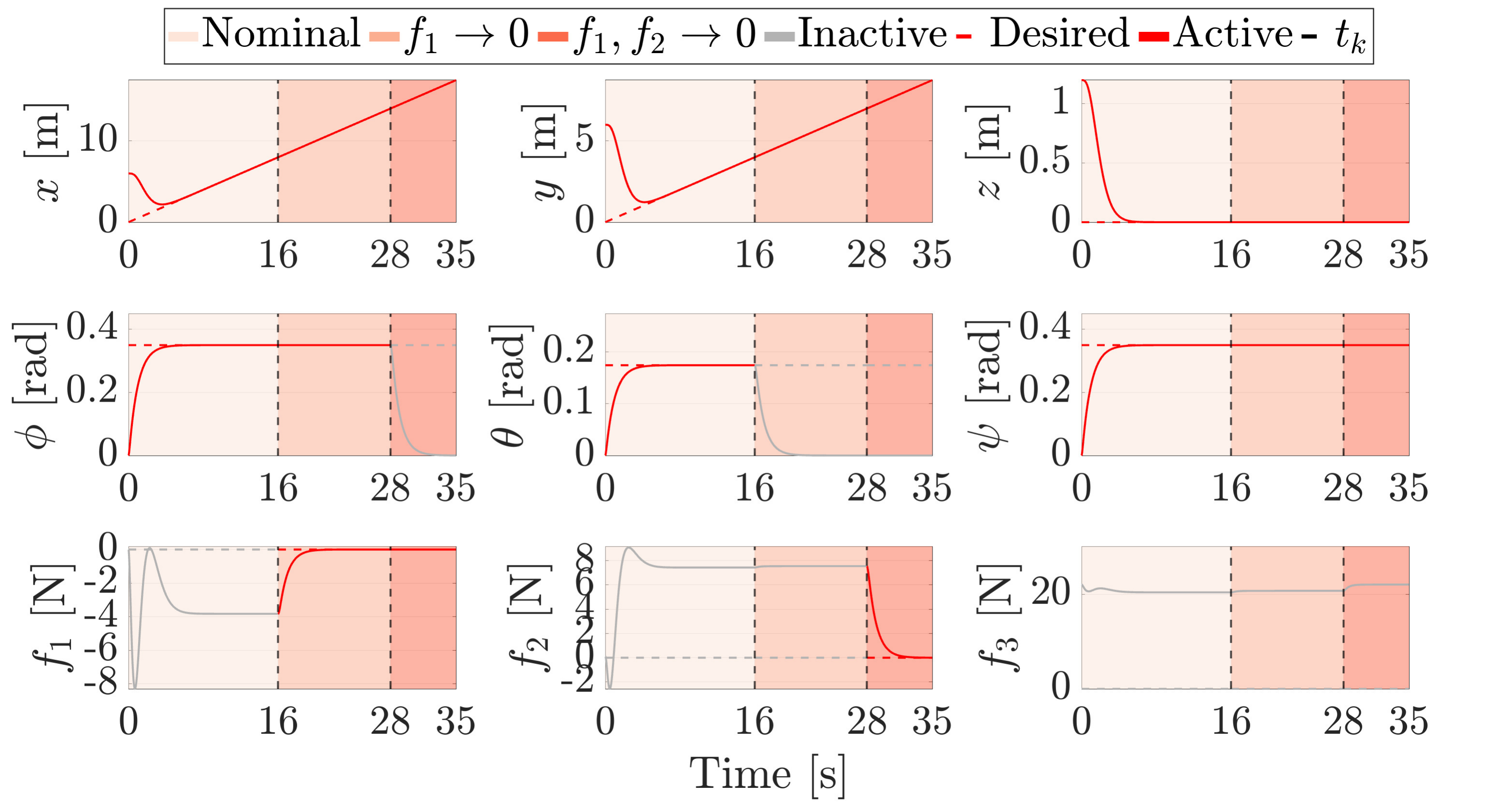}
    \caption{\textbf{Simulation.} Switching among three pairwise compatible melds on $\Sigma^{(\ell)}$: DF (\#2), FM (\#1), and QM (\#13). Dark gray: DF; light coral: FM; light gray: QM. No transients on shared outputs, consistent with Lemma~\ref{lem:zero_transient_common_ell} and Section~\ref{sec:switching-negotiable-common-ell}.}
    \label{fig:sim1_all}
\end{figure}

We simulate a 6-DoF fully actuated platform required to disable inputs over time, thereby changing its actuation mode. 
The aims are: (i) {robustness} to switching---seamless transitions on outputs shared across melds; (ii) exploitation of the meld structure---executing reduced tasks while keeping selected inputs at zero.

The (stacked) virtual input is chosen as
\begin{equation}
\label{eq:rb-w}
\wv =
\left[ \begin{smallmatrix}
\Km_{y,1}(\pv^d-\pv)+\Km_{y,2}(\dot\pv^d-\dot\pv)+\Km_{y,3}(\ddot\pv^{\,d}-\ddot\pv)+\Km_{y,4}(\dddot\pv^{\,d}-\dddot\pv)+\ddddot\pv^{\,d}\\[2pt]
\Km_{y,5}(\Phim^d - \Phim) + \Km_{y,6}(\dot{\Phim}^d - \dot{\Phim}) + \ddot{\Phim}^d\\[2pt]
-\Km_{u,7}\,\fv - \Km_{u,8}\,\dot{\fv}
         \end{smallmatrix}\right]
\end{equation}
with diagonal gains $\Km_{y,1} = 16\,\mathbf I$, $\Km_{y,2} = 32\,\mathbf I$, $\Km_{y,3} = 24\,\mathbf I$, $\Km_{y,4} = 8\,\mathbf I$, $\Km_{y,5} = 16\,\mathbf I$, $\Km_{y,6} = 8\,\mathbf I$, $\Km_{u,7} = 10\,\mathbf I$, $\Km_{u,8} = 10\,\mathbf I$ (dimensions consistent with the corresponding stacks). 
Initial conditions: $\pv(0)=\begin{bmatrix}6&6&1.2\end{bmatrix}^\top\,\si{\meter}$, $\dot\pv(0)=\mathbf 0\,\si{\meter\per\second}$, and $\Rm(0)=\mathbf I_{3\times 3}$. 
The reference is $\pv^d(t)=\big[\tfrac{t}{2}\ \ \tfrac{t}{4}\ \ 0\big]^\top\,\si{\meter}$ and fixed attitude $\psi^d=20^{\circ}$, $\theta^d=10^{\circ}$, $\phi^d=20^{\circ}$. 
The switching signal $\eta(t)$ moves along edges of $\mathcal{G}^\ell_\star$ with a fixed dwell time; as predicted by Section~\ref{sec:switching-negotiable-common-ell}, Fig.~\ref{fig:sim1_all} shows seamless negotiation between tasks with no transients on shared outputs. Starting from the Fully Actuated Mode (FM), where the flat output is the full pose $\yv$, the platform transitions at $t=16\,\mathrm{s}$ to a Dual-Force Mode (DF), where the active flat output becomes $[\pv^\top\;\phi\;\psi\; f_1]^\top$. 
It then transitions to the quadrotor mode (QM), where the flat output is $[\pv^\top\;\psi\;f_1\, f_2]^\top$.

\section{Application to a Mecanum Wheel Robot}\label{sec:mec_wheel}
In our prior work~\cite{2025a-MizGooFra}, we studied a mecanum-wheel  platform  with flat output $\yv=[x\;y\;\theta]^\top$. It was observed in~\cite{7496410} that sustained use of the lateral velocity input (denoted with $v_3$) can be energetically expensive, motivating the question of whether this input could be deactivated while preserving exact linearizability of a reduced task.
Although, the terminology of dexterity input was not introduced there, the analysis implicitly considered the 
 prolonged system with  pattern $\ell=\{1,0,1\}$ and exhibited two compatible flat outputs for the same prolonged system, namely $\yv$ and $[x\; y\; v_3]^\top$. A switching controller between these outputs which corresponds to the case $d_\ell = 2$ of the controller \eqref{eq:switching-law}  was designed. In the notation of the present paper, this establishes that the index corresponding to the lateral velocity input belongs to a dexterity subset of cardinality one, i.e., $\delta^3_{(\Sigma,\yv)} = 1.$

\section{Under-actuation and Fully-Actuated Systems}\label{sec:ua-fa}
Finally, after introducing the formal terminology and analyzing two representative mechanical platforms, we can articulate a unifying interpretation.

In both examples, we exhibited flat outputs that include some dexterity inputs alongside components of the original full-dimensional task. Each such selection corresponds to what is classically regarded as a different mechanical system. For instance, in the mecanum-wheel platform, the flat output $[x\; y\; v_3]^\top$ corresponds to the standard unicycle model~\cite{DELUCA2000687}, which can now be interpreted as the fully actuated system with the lateral input as part of a flat output. Likewise, for the fully actuated aerial platform, flat outputs containing  the lateral thrusts inputs recover the classical quadrotor model~\cite{2018c-FaeFraSca}; this too can be viewed as the same fully actuated system with selected lateral force inputs being part of a flat output.

More generally, these examples suggest a structural viewpoint: there is a single underlying system. Classical "underactuated" models~\cite{aneke2003control} arise as particular output selections of a suitably prolonged fully actuated system, corresponding to the removal (or regulation to zero) of dexterity inputs. Thus, different actuation modes are not distinct plants, but manifestations of a common prolonged dynamics under different flat-output configurations.
    
\section{Conclusions and Future Work}\label{subs:conclusion}

We presented a task-relative framework for nonlinear input-affine systems that (i) \emph{classifies} actuator inputs with respect to a given flat output $\yv$, and (ii) \emph{negotiates} between full and reduced tasks under actuator deactivation via a common dynamic prolongation. 
Relative to $(\Sigma,\yv)$ we introduced the taxonomy \emph{redundant / essential / dexterity} (Section~\ref{sec:io-class}, Defs.~\ref{defn:dextsubset}--\ref{defn_dext_ess}). 
Our main theoretical result (Theorem~\ref{thm:dexterity-input}) established an equivalence: a subset of inputs is \emph{dexterity} if and only if---on a nonempty intersection of validity sets---there exists a prolongation $\ell$ for which that subset can be embedded as a \emph{flat-input complement} in the output of the prolonged system.

As a consequence, whenever a common prolongation exists, the actuation modes obtained by (de)activating dexterity inputs can be interpreted as a \emph{single} prolonged plant endowed with different flat-output selections (Section~\ref{sec:common-prolongation}).
This yields a unified linearizing controller that switches along edges of the negotiability graph without inducing transients on the shared output components (Lemma~\ref{lem:zero_transient_common_ell}, Section~\ref{sec:switching-negotiable-common-ell}). 
A rigid-body flight case study (Section~\ref{sec:rigid-body}) classified thrust channels as dexterity inputs, listed representative melds, and demonstrated seamless downgrades among fully actuated, dual-force, and quad modes.

\subsubsection*{Limitations}
The approach relies on exact state-space linearization assumptions:  for a fixed prolongation, constant relative degree and a nonsingular decoupling matrix on the operating validity sets. 
It therefore inherits sensitivity to model uncertainty and parametric drift; compatibility requires nonempty intersection of validity sets; and switching uses a dwell-time condition. 
For Euler angles, attitude singularities restrict global coverage of validity sets.

\subsubsection*{Future directions}
\begin{itemize}
\item \emph{Robust/adaptive negotiation.} Incorporate model-uncertainty margins (e.g., ISS/robust linearization filters) and adaptive gain scheduling for the prolonged channels; synthesize dwell-time rules tied to closed-loop poles and perturbation bounds. 
\item \emph{Constraints and NMPC.} Embed input/state bounds and rate limits into a common-prolongation NMPC layer that preserves the zero-transient property on shared outputs.
\item \emph{Beyond fully flat outputs.} Extend the classification to partially feedback-linearizable systems and outputs with nonminimum-phase zeros; automate the computation of minimal prolongations $\ell$ using pure-prolongation conditions.
\item \emph{Global attitude and manifolds.} Replace Euler angles with representations on $\SO)$ (or quaternions) and generalize to tasks on manifolds with intrinsic validity-set characterization. 
\item \emph{Switching logic and verification.} Synthesize negotiability graphs from data, compute compatibility regions, and design formally verified switching policies (e.g., temporal-logic constraints) with safety certificates. 
\item \emph{Experimental validation.} Validate on fully actuated aerial platforms and multi-robot scenarios (cooperative load transport), including energy-aware scheduling of dexterity-input shutdowns.
\end{itemize}
Overall, the proposed equivalence and common-prolongation viewpoint unify actuation modes under one controller, offering principled, zero-transient task downgrades when selected inputs are lost or deliberately deactivated.

\bibliographystyle{IEEEtran}
\bibliography{Bib/bibAlias,Bib/bibAF,Bib/bibCustom,Bib/ref2}

\end{document}